\documentclass[reqno]{amsart}
\usepackage{mathrsfs}
\usepackage{graphicx}
\usepackage{amsmath}
\usepackage{amscd}
\usepackage{cite}
\usepackage{hyperref}
\usepackage{epsfig}
\usepackage{subfigure}
\usepackage{caption}
\usepackage{tikz}
  \usepackage{float}
  \usepackage{amssymb}
\usepackage{amsfonts}
\usepackage[all,cmtip]{xy}
\usepackage{appendix}
\usepackage{bm}
\usepackage{geometry}
\newtheorem{Theorem}{Theorem}[section]
\newtheorem{Example}[Theorem]{Example}
\newtheorem{Definition}[Theorem]{Definition}
\newtheorem{Proposition}[Theorem]{Proposition}
\newtheorem{Lemma}[Theorem]{Lemma}
\newtheorem{Corollary}[Theorem]{Corollary}

\newtheorem{Remark}[Theorem]{Remark}

\newcommand{\pd}{\partial}

\newcommand{\bC}{{\mathbb C}}

\newcommand{\bN}{{\mathbb N}}
\newcommand{\bP}{{\mathbb P}}

\newcommand{\bZ}{{\mathbb Z}}

\newcommand{\cB}{{\mathcal B}}

\newcommand{\cE}{{\mathcal E}}
\newcommand{\cF}{{\mathcal F}}

\newcommand{\cM}{{\mathcal M}}
\newcommand{\cO}{{\mathcal O}}
\newcommand{\cP}{{\mathcal P}}

\newcommand{\half}{\frac{1}{2}}

\newcommand{\wB}{{\widehat B}}
\newcommand{\wC}{{\widehat C}}
\newcommand{\wP}{{\widehat P}}

\newcommand{\tGa}{{\widetilde{\Gamma}}}

\newcommand{\be}{\begin{equation}}
\newcommand{\ee}{\end{equation}}
\newcommand{\bea}{\begin{eqnarray}}
\newcommand{\ben}{\begin{eqnarray*}}
\newcommand{\een}{\end{eqnarray*}}
\newcommand{\eea}{\end{eqnarray}}

\DeclareMathOperator{\Id}{id}

\usepackage{color}

\definecolor{yellow}{rgb}{1,1,0}
\definecolor{orange}{rgb}{1,.7,0}
\definecolor{red}{rgb}{1,0,0}
\definecolor{green}{rgb}{0,1,1}
\definecolor{white}{rgb}{1,1,1}

\definecolor{A}{rgb}{.75,1,.75}

\begin{document}

\newtheorem{myDef}{Definition}
\newtheorem{thm}{Theorem}
\newtheorem{eqn}{equation}

\title[Total partition function of local toric CY threefold]
{Total partition function with fermionic number fluxes\\of local toric Calabi--Yau threefold and\\KP integrability}

\author{Zhiyuan Wang}
\address{Zhiyuan Wang, School of Mathematics and Statistics,
	Huazhong University of Science and Technology,
	Wuhan, China}
\email{wangzy23@hust.edu.cn}

\author{Chenglang Yang}
\address{Chenglang Yang, Institute for Math and AI, Wuhan University, Wuhan, China}
\email{yangcl@pku.edu.cn}

\author{Jian Zhou}
\address{Jian Zhou, Department of Mathematical Sciences\\
	Tsinghua University, Beijing, China}
\email{jianzhou@mail.tsinghua.edu.cn}

\begin{abstract}

Aganagic, Dijkgraaf, Klemm, Mari\~{n}o and Vafa \cite{adkmv} predicted that
the open string partition function on a smooth toric Calabi--Yau threefold
should be a tau-function of multi-component KP hierarchy after considering the contributions
from nonzero fermion number fluxes through loops in the toric diagram.
In this paper,
we prove their prediction in the case of local toric Calabi--Yau threefolds.
More precisely,
we construct the total partition function of local toric Calabi--Yau threefolds using an operator on the fermionic Fock space which we developed in an earlier work \cite{wyz} to represent the topological vertex,
and show that the total partition function is the trace of an operator on the fermionic Fock space.
As an application,
we prove the KP integrability of the total partition function.
\end{abstract}

\maketitle


\section{Introduction}

In recent years,
integrable systems play an important role in the study of geometry of moduli spaces,
especially in Gromov--Witten theory.
The famous Witten Conjecture/Kontsevich Theorem \cite{wit,kon}
states that the generating function of intersection numbers of psi-classes
on the moduli spaces of stable curves is a tau-function of the KdV hierarchy.
Later more examples relating invariants in geometry to certain integrable hierarchies are discovered,
including the Hurwitz numbers \cite{pa, Ok1},
Gromov--Witten theory of $\mathbb{P}^1$ \cite{op},
$r$-spin models \cite{witten,fsz},
the Fan--Jarvis--Ruan--Witten theory \cite{fjr, lrz},
and the main subject of this paper --- the open Gromov--Witten invariants of
smooth toric Calabi--Yau threefolds \cite{adkmv, dz2, wyz},
etc.

In general Gromov--Witten invariants of a symplectic manifold is very difficult to compute.
When it admits some torus action,
one can consider the induced torus action on the moduli spaces and apply localization techniques.
This is especially efficient in  the case of smooth toric Calabi--Yau threefolds.
These are noncompact Calabi--Yau threefolds which admits construction via toric geometry.
In particular,
such spaces can be described by some combinatorial data called toric diagrams.
For examples,
local $\bP^2$ means the total space of the canonical line bundle $\cO_{\bP^2}(-3)$,
and local  $\bP^1\times \bP^1$ means $\cO_{\bP^1\times \bP^1}(-2,-2)$,
and their toric diagrams are listed in Figure \ref{fig1}.
\begin{figure}[H]
\begin{tikzpicture}[scale=0.5]
\draw [thick] (0,0) -- (0,2);
\draw [thick] (0,0) -- (2,0);
\draw [thick] (2,0) -- (0,2);
\draw [thick] (0,0) -- (-0.9,-0.9);
\draw [thick] (0,2) -- (-0.6,3.2);
\draw [thick] (2,0) -- (3.2,-0.6);
\node [align=center,align=center] at (1,-1.4) {local $\bP^2$};
\draw [thick] (8,0) -- (10,0);
\draw [thick] (8,0) -- (8,2);
\draw [thick] (8,2) -- (10,2);
\draw [thick] (10,0) -- (10,2);
\draw [thick] (8,0) -- (7.1,-0.9);
\draw [thick] (10,0) -- (10.9,-0.9);
\draw [thick] (8,2) -- (7.1,2.9);
\draw [thick] (10,2) -- (10.9,2.9);
\node [align=center,align=center] at (9,-1.4) {local $\bP^1\times \bP^1$};
\end{tikzpicture}
\caption{Toric diagrams for local $\bP^2$ and local $\bP^1\times \bP^1$}
\label{fig1}
\end{figure}
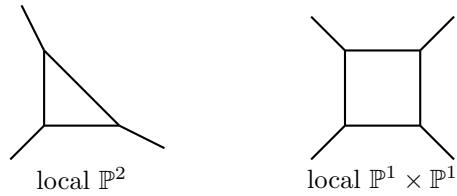

For a physicist,
such diagrams look very similar to the Feynman diagrams in quantum field theory.
Amazingly,
string theorists discovered a method to compute the Gromov--Witten invariants of such spaces
by suitable Feymann rules.
This is the notion of the topological vertex first proposed by physicists
Aganagic, Klemm, Mari\~{n}o, and Vafa in their study of
the connection between open string theory and large $N$ Chern--Simons theory \cite{akmv}.
The topological vertex is the generating series of open Gromov--Witten invariants of $\bC^3$,
and the open Gromov--Witten invariants of a general toric CY threefold can be obtained by certain gluing procedures of the topological vertex
by assigning suitable weights to the internal edges.
This is slightly different from the Feynman rules in quantum field theory.
Even though the toric diagrams may have some loops,
no integrals of loop momenta are involved when we compute the (open) Gromov--Witten invariants.
See \cite{lllz,moop} for mathematical theories of the topological vertex.

The topological vertex not only provides a computational tool,
but also can be used to establish the relationship between the open Gromov--Witten theory of a general toric Calabi--Yau threefold
and the theory of integrable hierarchies.
In \cite{adkmv}, Aganagic--Dijkgraaf--Klemm--Mari\~{n}o--Vafa proposed that one can reconstruct the topological vertex via certain fermionic field on a Riemann surface.
Based on the fermionic picture,
they conjectured that the topological vertex has a fermionic representation
as a Bogoliubov transform
of the fermionic vacuum vector in the $3$-component fermionic Fock space,
where the coefficients are given by some explicit closed formulas.
As a consequence,
the topological vertex provides a tau-function of the $3$-component KP hierarchy,
and it should be able to  generalize  this integrability  to a general Calabi--Yau threefold.
(For the definition of the multi-component KP hierarchy,
see the paper of Kac and van de Leur \cite{kv}.)
This representation of the topological vertex was called the ADKMV Conjecture
in \cite{dz1, dz2, wyz, kas}.
Recently,
this conjecture and its framed version \cite{dz1} were proved by the authors in \cite{wyz},
and a straightforward consequence of this result
is that the topological vertex provides a tau-function of the $3$-component KP hierarchy.
As an application,
one easily sees the multi-component KP integrability
of the open Gromov--Witten theory of a toric Calabi--Yau threefold whose
toric diagram is a tree by the gluing rule \cite{adkmv, dz2, kas, wyz}.

When the toric diagram of a toric Calabi--Yau threefold $X$ contains loops,
the problem is more involved.
In this case the open string amplitudes do not directly correspond to a tau-function of
the multi-component KP hierarchy
since it only contains the contributions of zero fermion number sector,
see \cite[\S 4.7]{adkmv}.
From a physical viewpoint of local mirror symmetry,
the mirror curve is a sort of fatting of the toric diagram.
When the toric diagram has loops,
the mirror curve has handles.
Aganagic {\em et al} introduced $g$ new variables $\theta_1,\dots, \theta_g$ when the toric diagrams has $g$ independent loops,
and considered the total partition function:
\be  \label{eqn:Charge}
Z(\theta_i) = \sum_{N_1,\cdots,N_g\in \bZ} Z_{N_1,\dots,N_g} \cdot \exp(\sum_{i=1}^g iN_i\theta_i)
\ee
 where $Z_{N_1,\dots,N_g}$ should be the contribution of the sector with fermion number flux $N_i$ through the $i$-th handle.
 Their motivation is that the total partition function should  still be represented in the fermionic picture as a Bogoliubov transform,
 and hence it gives rise to a tau-function of the multicomponent KP hierarchy.
 In  \cite{dz2}, the approach of   using the gluing the topological vertex in the fermionic picture was taken to verify \eqref{eqn:Charge}.

 As pointed out in \cite{adkmv},
 $Z_{N_1,\dots,N_g}$ may not contain more information than $Z_{N_1=0, \dots, N_g=0}$.
 It was argued that  $Z_{N_1,\dots,N_g}$ should be determined from $Z_{N_1=0, \dots, N_g=0}$
 by a shift in the K\"ahler parameters, more precisely   \cite[(4.17)]{adkmv}:
 \be  \label{eqn:Shift}
 Z_{N_1, \dots, N_g}(S_1, \dots, S_g) = Z_{N_1=0, \dots, N_g=0}(S_i + N_im_ig_s),
 \ee
 where $m_i$ are some geometrically defined intersection numbers,
 and $g_s$ is the string coupling constant.
It would be very interesting to study \eqref{eqn:Shift}
using the approach of  fermionic gluing principle in \cite{dz2},
but such a treatment has not yet appeared.

In this paper,
we will present a new approach that simultaneously deals with both \eqref{eqn:Charge} and \eqref{eqn:Shift}.
We will focus on the case of one-loop toric diagrams,
in which \eqref{eqn:Charge} takes the form
\be
\label{eq-intro-Ztotal}
Z^{\text{total}} (\mathbf{t};Q_i,\gamma_i;\Xi) =
\sum_{N\in \bZ} Z^{(N)}(\mathbf{t}; Q_i,\gamma_i)\cdot \Xi^{-N}.
\ee
Surprisingly, we find that \eqref{eqn:Shift} needs some modification to be true.
The modified formula turns to be
\begin{equation}  \label{eqn-N-0-0}
Z^{(N)} (\bm 0; Q_i,\gamma_i) =
(Q_1\cdots Q_M)^{\frac{N^2}{2}} \cdot (-1)^{\gamma \frac{N^2}{2}} \cdot q^{(\gamma+2M)\frac{N(4N^2 -1)}{24}}
\cdot Z (\bm 0; q^{(\gamma_i+2)N}Q_i,\gamma_i),
\end{equation}
and more generally,
\be  \label{eqn-N-0-t}
 Z^{(N)} (\mathbf{t};Q_i,\gamma_i)
= q^{-NL_0} \cdot (Q_1\cdots Q_M)^{\frac{N^2}{2}} (-1)^{\gamma \frac{ N^2}{2}}
q^{(\gamma +2M) \frac{N(4N^2 -1)}{24}}
Z(\mathbf{t}; q^{(\gamma_i+2)N}Q_i,\gamma_i),
\ee
See \S \ref{sec:def Ztotal} for detailed explanations of the notations.
Note that
even though $Z^{(N)} (\bm 0; Q_i,\gamma_i) $ and $Z^{(N)} (\mathbf{t};Q_i,\gamma_i)$  come out of a fermionic approach
and they do not seem to have direct geometric interpretations in terms of Gromov--Witten invariants,
the equations \eqref{eqn-N-0-0} and  \eqref{eqn-N-0-t} give them such interpretations in an indirect way.
The factor on the right-hand sides of these equation do not seem to have been predicted by physicists,
and they deserve further studies.

Here we are considering a formal toric Calabi--Yau threefold  $X$  whose toric diagram is a polygon:
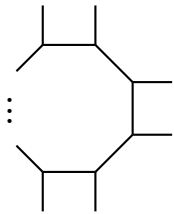
\begin{equation*}
\begin{tikzpicture}[scale=0.35]
\draw [thick] (-1,0) -- (1,0);
\draw [thick] (1,0) -- (2.414,-1.414);
\draw [thick] (1,0) -- (1,1.5);
\draw [thick] (2.414,-1.414) -- (2.414,-3.414);
\draw [thick] (1,-4.828) -- (2.414,-3.414);
\draw [thick] (1,-4.828) -- (-1,-4.828);
\draw [thick] (1,-4.828) -- (1,-6.328);
\draw [thick] (2.414,-1.414) -- (4.414-0.5,-1.414);
\draw [thick] (2.414,-3.414) -- (4.414-0.5,-3.414);
\draw [thick] (-1,0) -- (-2,-1);
\draw [thick] (-1,-4.828) -- (-2,-3.828);
\draw [thick] (-1,0) -- (-1,1.5);
\draw [thick] (-1,-4.828) -- (-1,-6.328);
\node [align=center,align=center] at (-2.25,-2.2) {{\bf $\vdots$}};
\end{tikzpicture}
\end{equation*}
The combinatorial data needed to apply the gluing rule of topological vertices
are some K\"ahler parameters $Q_i$ and framing factors $\gamma_i$ for $i=1,\cdots,M$,
where $M$ is the number of vertices in the toric diagram.
In particular,
if $X$ is the total space of the canonical bundle on a non-singular toric Fano surface,
then the toric diagram associated with $X$ is of this form,
and the K\"ahler parameters correspond the size of invariant curves
while the framing factors correspond to the self-intersection numbers of invariant rational curves in the toric surface.
We use \eqref{eqn-N-0-t} to construct a total partition function of the open string theory on $X$.
For simplicity we will restrict to the case of one outer brane in this work,
and the result can be directly generalized to the general case of several outer branes and
multi-component KP hierarchy.
More precisely, we have:
\begin{Theorem}
\label{thm-intro-1}
Denote by $Z(\mathbf{t};Q_i,\gamma_i)$ the partition function of open string amplitudes of $X$ with one outer brane,
where $X$ is a formal Calabi--Yau manifold whose toric diagram consists of a polygon together with
an external leg attached to each vertex,
and $\gamma_i$, $Q_i$ are the framing factors and K\"ahler parameters respectively.
Then the total partition function
\be
\label{eq-intro-Ztotal}
Z^{\text{total}} (\mathbf{t};Q_i,\gamma_i;\Xi) =
\sum_{N\in \bZ} Z^{(N)}(\mathbf{t}; Q_i,\gamma_i)\cdot \Xi^{-N},
\ee
a tau-function of the KP hierarchy, and
$$Z^{(0)}(\mathbf{t};Q_i,\gamma_i)=Z(\mathbf{t};Q_i,\gamma_i).$$
\end{Theorem}

The details of the construction of these formal functions $\{Z^{(N)}(\mathbf{t};Q_i,\gamma_i)\}_{N\in\mathbb{Z}}$
will be given in \S \ref{sec:def Ztotal}.
It is well-known that a formal power series tau-function of the KP hierarchy corresponds to
a point on an semi-infinite dimensional Grassmannian called Sato Grassmannian,
due to the famous Sato theory of integrable systems \cite{djm, sa}.
We will write down a formula for the affine coordinates of the point corresponding to
the tau-function $Z^{\text{total}} (\mathbf{t};Q_i,\gamma_i;\Xi)$
on the big cell of the Sato Grassmannian,
see \S \ref{sec-KPaffine} for details.

We apply the above results to the case of the total space $X$ of the canonical line bundle on a smooth toric Fano surface $S$.
In this case $s_i = -\log Q_i$ and $\gamma_i$ are
the K\"ahler parameters and self-intersection numbers of invariant rational curves in $S$ respectively.
We have:
\begin{Corollary}\label{cor:main ZS}
Let $X$ be the total space of the canonical bundle on a non-singular  toric Fano surface $S$,
and denote by $Z_S(\mathbf{t})$ the generating series of the open Gromov--Witten invariants of $X$
with one outer brane.
Then $Z_S(\mathbf{t})$ can be extended into a tau-function
$$Z_S^{\text{total}}(\mathbf{t};\Xi)
=\sum_{N\in \bZ} Z^{(N)}_S(\mathbf{t}) \cdot \Xi^{-N}$$ of the KP hierarchy such that $Z^{(0)}_S(\mathbf{t})=Z_S(\mathbf{t})$.
\end{Corollary}

A tau-function $\tau(\mathbf{t})$ of the KP hierarchy
corresponding to a point on the big cell of the Sato Grassmannian
is uniquely determined by its affine coordinates and the constant term $\tau(\bm 0)$.
In a special case,
we are able to carry out the concrete computations and write down
the explicit formulas for the constant term and affine coordinates
of the total partition function,
see \S \ref{sec-eg} for details.
Take $\gamma_1=\cdots =\gamma_M =-2$,
and denote by $Z^{\text{total}} (\mathbf{t}; Q_i,-2; \Xi)$ the total partition function of this special model.
(We are inspired by Kulikov elliptic surface singularity of type $I_{M-1}$ whose exceptional divisor of the minimal resolution
is  a cycle of $(-2)$-curves \cite[p. 294]{nem}.)
In this case we have:
\begin{Theorem} \label{thm-Modular}
The constant term of the total partition function $Z^{\text{total}} (\mathbf{t}; Q_i,-2; \Xi)$ is:
\begin{equation*}
\begin{split}
Z^{\text{total}} (\bm 0; Q_i,-2; \Xi)= &
\Theta_3 \big( (-1)^M \cdot \Xi^{-1}; Q^{1/2} \big) \cdot
\prod_{i=1}^\infty \frac{1}{1-Q^i} \\
&\cdot
\prod_{1\leq k <l \leq M}  \frac{1}{\cM (\prod_{i=k}^{l-1}Q_i ; q)} \cdot
\prod_{k,l= 1}^M \prod_{j=0}^\infty
\frac{1}{\cM(Q^j \cdot \prod_{a=1}^{k-1} Q_a \cdot \prod_{b=l}^M Q_i ;q)},
\end{split}
\end{equation*}
where $Q = Q_1Q_2\cdots Q_M$, and
\begin{equation*}
\Theta_3 (t;q) = \sum_{n\in \bZ} q^{n^2}  t^n,
\qquad\qquad
\cM(z;q) = \prod_{n=1}^\infty (1-z q^{-n})^n.
\end{equation*}
Moreover,
the affine coordinates of the tau-function $Z^{\text{total}} (\mathbf{t}; Q_i,-2; \Xi)$ on the big cell of
the Sato Grassmannian are:
\begin{equation*}
\begin{split}
a_{n,m}^{(-2)} =
& \frac{\Theta_3 \big( (-1)^M (q^{m+n+1}  \Xi)^{-1} ; Q^{\half} \big)}
{\Theta_3 \big( (-1)^M  \Xi^{-1}; Q^{\half} \big) }
\cdot
\frac{(-1)^{m+1} q^{\half m^2 +m +\half n +\half} }
{(1-q^{m+n+1}) \prod_{j=0}^{m-1} (1-q^{m-j}) \cdot \prod_{i=0}^{n-1}(1-q^{n-i}) } \\
&
\cdot \prod_{j= 0}^\infty \prod_{l=2}^M
\frac{  1-\prod_{a=1}^{l-1}Q_a \cdot q^{-n-j-1}}
{1-\prod_{a=1}^{l-1}Q_a \cdot q^{m-j}}
 \cdot \prod_{j=1}^\infty \frac{(1-Q^j)^2}{(1-Q^j q^{m+n+1})(1-Q^j q^{-m-n-1})}
\\
& \cdot
\prod_{j=0}^\infty \prod_{l=1}^M \prod_{i\geq 0} \frac{1-Q^j \prod_{b=l}^M Q_b \cdot q^{-m-i-1}}
{1-Q^j \prod_{b=l}^M Q_b \cdot q^{n-i}}
\cdot
\prod_{j=1}^\infty\prod_{k=1}^M \prod_{i\geq 0} \frac{1-Q^j  \prod_{a=1}^{k-1} Q_a \cdot q^{-n-i-1}}
{1-Q^j \prod_{a=1}^{k-1} Q_a \cdot q^{m-i}}.
\end{split}
\end{equation*}
\end{Theorem}

We also derive a difference equation for the fermionic one-point function of the total partition function,
which is the first basis vector
of the corresponding point on the Sato Grassmannian.
The fermionic one-point function is also a principal specialization of the tau-function,
and we understand the difference equation as the quantum spectral curve associated with this tau-function
following the ideas in \cite{gs, al}.
\begin{Proposition}
Let
$\Psi (z) =
1+ \sum_{m\geq 0}  a_{0,m}^{(-2)} \cdot z^{-m-1}$
be the first basis vector
of the point on the Sato Grassmannian corresponding to the tau-function $Z^{\text{total}} (\mathbf{t}; Q_i,-2; \Xi)$,
then:
\begin{equation*}
\wP \Psi(z) =0,
\end{equation*}
where $\wP$ is the $q$-difference operator given by \eqref{eq-def-oprP}.
\end{Proposition}

We will also present some data for the total partition function of the local $\mathbb{P}^2$ model in \S \ref{sec:P2}.
We compute the leading terms of the corresponding total partition function,
verify the Hirota bilinear relation,
and compare with the data in the literature.

Let us say a few words about the method we use to establish these results.
In the literature the topological vertex in the fermionic picture is usually represented as
a state in $\cF \otimes \cF \otimes \cF$,
where $\cF$ is the fermionic Fock space.
In this work we are based on our new approach to the fermionic representation of topological vertex developed in our earlier work \cite{wyz}.
To each partition $\mu$,
we assign an operator $\Psi_\mu$  on $\cF$,
such that for any partitions $\lambda$ and $\nu$
we have
\ben
&&	W_{\lambda,\mu,\nu} (q) =
	\langle \lambda |  \Psi_{\mu}(q) q^{-K /2} |\nu^t \rangle,
	\een
 where $W_{\lambda,\mu,\nu}$ is the topological vertex in the bosonic picture,
$\nu^t$ is the conjugate partition of $\nu$, and $K$ is the cut-and-join operator.
Note the usual fermionic representation of the topological vertex needs three copies of the fermionic Fock space,
but working with $\Psi_\mu$ one only needs one copy.
As a consequence, in the usual approach, for a toric one-loop diagram with $M$ vertices,
one needs to work with $3M$ copies of fermionic Fock spaces, $3$ copies for each vertex;
but in our approach,
we only need to work with $1$ copy of $\cF$.
More precisely,  we have:
\be
\label{eq-ferm-Ztotal}
\begin{split}
&Z^{\text{total}} (\mathbf{t}; Q_i,\gamma_i; \Xi)
= \sum_{N\in \bZ} \sum_{\lambda,\mu \in \cP}  s_\lambda(\mathbf{t})\cdot
\Big\langle \mu \Big|
R^N \Xi^C
\Big( \Psi_\lambda (q) q^{-\half (\gamma_1+2)K} (-1)^{\gamma_1  L} Q_1^L \Big)\cdot  \\
& \quad
\Big( \Psi_{(\emptyset)}(q) q^{-\half (\gamma_2+2) K}
(-1)^{\gamma_2 L} Q_2^L \Big) \cdots
\Big( \Psi_{(\emptyset)}(q) q^{-\half (\gamma_M+2) K}
(-1)^{\gamma_M L} Q_M^L \Big)
R^{-N} \Big| \mu \Big\rangle .
\end{split}
\ee
See \S \ref{sec:def Ztotal} for notations.
This is what we use to prove Theorem \ref{thm-intro-1}.
For the KP integrability,
we prove a determinantal formula in Theorem \ref{thm:Determinantal}.
This formula is proved by a generalization of the Wick's formula at finite temperature
in the literature (c.f. Lemma \ref{lem-det}).
By the determinantal formula,
we express the constant term and the affine coordinates of $Z^{\text{total}} (\mathbf{t}; Q_i,\gamma_i; \Xi)$
in terms of traces on $\cF$ similar to the right-hand side of \eqref{eq-ferm-Ztotal}.
These traces can be explicitly evaluated in the case of $\gamma_1 = \cdots = \gamma_M = -2$
by commuting the operators, thus proving Theorem \ref{thm-Modular}.

The rest of this paper is organized as follows.
In \S \ref{sec:pre},
we review the Schur functions, free fermions,
and basic materials for the boson-fermion correspondence and KP hierarchy.
In \S \ref{sec-TV},
we review the topological vertex formalism for
the open GW theory of toric CY threefolds.
In \S \ref{sec:def Ztotal} we give a construction of the contributions of nonzero fermion number fluxes
on local toric CY threefolds,
and use them to define the total partition function of open string amplitudes.
Later in \S \ref{sec:KP integrability},
we prove the KP integrability of the total partition function
(Theorem \ref{thm-intro-1} and Corollary \ref{cor:main ZS}).
Finally in \S \ref{sec-eg} and \S \ref{sec:P2},
we study two special models mentioned above,
i.e., the $(-2)$-model and the local $\bP^2$ model.

\section{Preliminaries}
\label{sec:pre}
In this section we recall some preliminaries of the boson-fermion correspondence
and the KP hierarchy \cite{djm, sa}.

\subsection{Schur functions and skew Schur functions}

First we recall the  Schur functions and skew Schur functions.
See e.g. \cite[\S I]{mac} for details.

Let $\lambda=(\lambda_1,\lambda_2,\cdots,\lambda_{l})$
(where $\lambda_1\geq \cdots\geq \lambda_{l}> 0$) be a partition of an integer.
We denote
\be
|\lambda| = \lambda_1+ \lambda_2 + \cdots +\lambda_l.
\ee
The number $l = l(\lambda)$ is called the length of this partition.
Here we allow the empty partition $(\emptyset)$ of $0$,
and we denote by $\mathcal{P}$ the set of all partitions.
The transpose $\lambda^t=(\lambda_1^t,\cdots,\lambda_m^t)$ of $\lambda$ is the partition of $n$ defined by
$m=\lambda_1$ and $\lambda_j^t = |\{i| \lambda_i\geq j\}|$.
The Young diagram associated with $\lambda^t$
is obtained by flipping the Young diagram associated with $\lambda$ along the diagonal.
The Frobenius notation of a partition $\lambda$ is defined to be:
\begin{equation*}
\lambda= (m_1,m_2,\cdots,m_k | n_1,n_2,\cdots,n_k),
\end{equation*}
where $k$ is the number of boxes on the diagonal of the corresponding Young diagram,
and $m_i = \lambda_i - i$, $n_i=\lambda_i^t - i$ for $i=1,2,\cdots,k$.
Let $\lambda = (m_1,\cdots, m_k|n_1,\cdots, n_k)$ be a partition of integer,
and we denote:
\be
\label{eq-def-kappamu}
\kappa_\mu = \sum_{i=1}^{l(\mu)} \mu_i(\mu_i -2i +1)
=\sum_{i=1}^k (m_i+\half)^2 - \sum_{i=1}^k (n_i+\half)^2.
\ee
In particular, one has $\kappa_{(\emptyset)} =0$.

Consider the partition $\lambda=(m|n) = (m+1,1^n)$ whose Young diagram is a hook,
and in this case the Schur function $s_\lambda = s_{(m|n)}$ is defined by:
\begin{equation*}
s_{(m|n)}= h_{m+1}e_n - h_{m+2}e_{n-1} + \cdots
+ (-1)^n h_{m+n+1},
\end{equation*}
where $h_n$ and $e_n$ are the
complete symmetric function and elementary symmetric function of degree $n$ respectively.
For a general $\lambda=(m_1,\cdots,m_k|n_1,\cdots,n_k)$,
the Schur function $s_\lambda$ is defined by:
\begin{equation*}
s_\lambda  = \det (s_{(m_i|n_j)} )_{1\leq i,j\leq k}.
\end{equation*}
The Schur function indexed by the empty partition is defined to be $s_{(\emptyset)} = 1$.
Denote by $\Lambda$ the infinite-dimensional linear space of all symmetric functions,
then Schur functions $\{s_\lambda\}_{\lambda\in\mathcal{P}}$ form a basis of  $\Lambda$.
One can define a natural inner product $(\cdot,\cdot)$ on $\Lambda$ by
taking $\{s_\lambda\}_{\lambda\in\mathcal{P}}$ to be an orthonormal basis.

Let $p_n\in \Lambda$ be the Newton symmetric function of degree $n$.
It is well-known that $\{p_\lambda\}$ gives another basis for $\Lambda$,
where $p_\lambda = p_{\lambda_1}\cdots p_{\lambda_l}$ for a partition $\lambda = (\lambda_1,\cdots,\lambda_l)$.
The above two bases for $\Lambda$ are related by (see e.g. \cite{mac}):
\be
\label{eq-newton-schur}
s_\mu = \sum_{\lambda} \frac{\chi^\mu (C_\lambda)}{z_\lambda} p_\lambda,
\ee
where $\chi^\mu$ is the character of the symmetric group $S_{|\mu|}$ indexed by the partition $\mu$,
and $C_\lambda \subset S_{|\mu|}$ is the conjugacy class whose cycle type is $\lambda$.
The summations on the right-hand sides of \eqref{eq-newton-schur} are over all partitions $\lambda$
with $|\lambda|=|\mu|$.
Sometimes we regard $s_\mu$ as a function in the variables $\mathbf{t} = (t_1,t_2,\cdots)$ via \eqref{eq-newton-schur}
where $t_n = \frac{p_n}{n}$,
and call it the Schur polynomial indexed by $\mu$.
In fact,
if one defines a degree by setting $\deg (t_n) = n$,
then it is easy to see that $s_\mu (\mathbf{t})$ is a weighted homogeneous polynomial of degree $|\mu|$ in
the variables $t_1,t_2,\cdots,t_{|\mu|}$.

Finally we recall the definition of the skew Schur functions.
Let $\lambda,\mu$ be two partitions,
and then the skew Schur function $s_{\lambda/\mu} \in \Lambda$ is determined by:
\begin{equation*}
(s_{\lambda/\mu},s_\nu ) = (s_\lambda, s_\mu s_\nu),
\qquad \forall \text{ partition $\nu$},
\end{equation*}
where $(\cdot,\cdot)$ is the inner product defined above.
Or equivalently,
\begin{equation*}
s_{\lambda/\mu} = \sum_\nu c_{\mu\nu}^\lambda s_\nu
\end{equation*}
where $\{c_{\mu\nu}^\lambda\}$ are the Littlewood-Richardson coefficients:
\begin{equation*}
s_\mu s_\nu = \sum_{\lambda} c_{\mu\nu}^\lambda s_\lambda,
\end{equation*}

\subsection{Free fermions and the fermionic Fock space}
\label{sec-freeferm}

In this subsection we recall the action of free fermions on the fermionic Fock space.
See \cite{djm, kac, sa} for details.

Let $\bm{a}=(a_1,a_2,\cdots)$ be a sequence of half-integers $ a_i \in \bZ+\half$ satisfying $a_1<a_2<a_3<\cdots$.
The sequence $\bm{a}$ is said to be admissible if:
\begin{equation*}
\big|(\bZ_{\geq 0}+\half)-\{a_1,a_2\cdots\}\big|<\infty,
\qquad
\big|\{a_1,a_2\cdots\}-(\bZ_{\geq 0}+\half)\big|<\infty.
\end{equation*}
One can define a semi-infinite wedge product $|\bm a\rangle$ for every admissible $\bm{a}$:
\be
\label{eq-basiscF-a}
| \bm a\rangle =
z^{a_1} \wedge z^{a_2} \wedge z^{a_3} \wedge \cdots
\ee
The fermionic Fock space $\cF$ is the infinite-dimensional linear space consisting of
all formal (infinite) summations of the following form:
\begin{equation*}
\sum_{\bm a:\text{ admissible}} c_{\bm a} |\bm a\rangle,
\end{equation*}
where $c_{\bm a}$ are some complex numbers.

Let $\bm a$ be admissible,
then the charge of the vector $|\bm a\rangle \in \cF$ is defined to be:
\begin{equation*}
\text{charge}(|\bm a\rangle)=
\big|\{a_1,a_2\cdots\}-(\bZ_{\geq 0}+\half)\big|
-\big|(\bZ_{\geq 0}+\half)-\{a_1,a_2\cdots\}\big|,
\end{equation*}
and this gives a decomposition of the fermionic Fock space:
\begin{equation*}
\cF=\bigoplus_{n\in \bZ} \cF^{(n)},
\end{equation*}
where $\cF^{(n)}$ is the subspace spanned by homogeneous vectors of charge $n$.
We will denote:
\begin{equation*}
|n\rangle = z^{n+\half}\wedge z^{n+\frac{3}{2}}
\wedge z^{n+\frac{5}{2}} \wedge \cdots \in \cF^{(-n)}.
\end{equation*}
In particular,
$|0\rangle = z^{\half}\wedge z^{\frac{3}{2}}
\wedge z^{\frac{5}{2}} \wedge \cdots \in \cF^{(0)}$
is called the fermionic vacuum vector.

The subspace $\cF^{(0)} \subset \cF$ of charge zero has a natural basis indexed by partitions of integers.
For a partition $\mu=(\mu_1,\cdots,\mu_l)$,
we use the notation $\mu_{l+1}=\mu_{l+2}=\cdots=0$
and denote:
\be
\label{eq-basisF-mu}
|\mu\rangle =
z^{\frac{1}{2}-\mu_1}\wedge z^{\frac{3}{2}-\mu_2}\wedge
z^{\frac{5}{2}-\mu_3}\wedge \cdots \in\cF^{(0)}.
\ee
Then the vectors $\{|\mu\rangle\}_{\mu\in\mathcal{P}}$ form a basis for $\cF^{(0)}$.
Here the empty partition $(\emptyset)$ of $0\in \bZ$
corresponds to the vacuum vector $|0\rangle\in\cF^{(0)}$.

The free fermions $\{\psi_r, \psi_r^*\}_{r\in \bZ+\half}$ are a family of operators
acting on $\cF$ by:
\begin{equation*}
\psi_r |\bm a\rangle = z^r \wedge |\bm a\rangle,
\qquad \forall r\in \bZ+\half,
\end{equation*}
and
\begin{equation*}
\psi_r^* | \bm a\rangle=
\begin{cases}
(-1)^{k+1} \cdot z^{a_1}\wedge z^{a_2}\wedge \cdots \wedge \widehat{z^{a_k}} \wedge \cdots,
&\text{if $a_k = -r$ for some $k$;}\\
0, &\text{otherwise.}
\end{cases}
\end{equation*}
It is clear that the operators $\{\psi_r\}$ all have charge $1$,
and $\{\psi_r^*\}$ all have charge $(-1)$.
They satisfy the following anti-commutation relations:
\be
\label{eq-anticomm-psi}
[\psi_r,\psi_s]_+=0,
\qquad
[\psi_r^*,\psi_s^*]_+=0,
\qquad
[\psi_r,\psi_s^*]_+= \delta_{r+s,0}\cdot \Id_\cF,
\ee
where the bracket is defined by $[a,b]_+ =ab+ba$.
In other words,
these fermions generate a Clifford algebra.

The operators $\{\psi_r,\psi_r^*\}_{r<0}$ are called the fermionic creators,
and $\{\psi_r,\psi_r^*\}_{r>0}$ are called the fermionic annihilators,
since one can easily check that:
\begin{equation*}
\psi_r |0\rangle=0,
\qquad
\psi_r^* |0\rangle=0,
\qquad \forall r>0.
\end{equation*}
Moreover,
every element of the form $|\mu\rangle$ (where $\mu$ is a partition)
can be obtained by applying fermionic creators successively to the vacuum vector $|0\rangle$:
\be
\label{eq-ferm-basismu}
|\mu\rangle=(-1)^{n_1+\cdots+n_k}\cdot
\psi_{-m_1-\half} \psi_{-n_1-\half}^* \cdots
\psi_{-m_k-\half} \psi_{-n_k-\half}^* |0\rangle,
\ee
and here $\mu=(m_1,\cdots, m_k | n_1,\cdots,n_k)$ in Frobenius notation.

One can define an inner product $(\cdot,\cdot)_{\cF}$ on $\cF$ by taking
$\{|\bm a\rangle | \text{$\bm a$ is admissible}\}$ to be an orthonormal basis,
and then $\{|\mu  \rangle \}_\mu$ is an orthonormal basis for $\cF^{(0)}$.
It is clear that
$\psi_r$ and $\psi_{-r}^*$ are adjoint to each other with respect to this inner product.
Given two admissible sequences $\bm a$ and $\bm b$,
we will denote $\langle \bm b | \bm a\rangle =(|\bm a\rangle,|\bm b\rangle)_{\cF}$.

Let $A$ be an operator on the Fock space $\cF$,
and we will denote by
$\langle A \rangle = \langle 0| A |0\rangle$
the vacuum expectation value of $A$.
Let $w_1, w_2,\cdots,w_k$ be some linear combinations of the free fermions $\{\psi_r,\psi_r^*\}$,
then the vacuum expectation value of $w_1w_2\cdots w_k$ can be computed using Wick's Theorem,
see e.g. \cite[\S 4.5]{djm} for details.

\subsection{Tau-functions of the KP hierarchy and affine coordinates}
\label{pre-KPtau}

Now we recall the boson-fermion correspondence
and Sato's theory of the KP hierarchy \cite{djm, sa},
including the notion of affine coordinates of a tau-function.

For $n\in \bZ$,
denote by $\alpha_n$ the following operator on $\cF$:
\be
\alpha_n = \sum_{s\in \bZ+\half} :\psi_{-s} \psi_{s+n}^*:.
\ee
Here $: \psi_{-s} \psi_{s+n}^* :$ denotes the normal-ordered product of fermions:
\begin{equation*}
:\phi_{r_1}\phi_{r_2}\cdots \phi_{r_n}:
=(-1)^\sigma \phi_{r_{\sigma(1)}}\phi_{r_{\sigma(2)}}\cdots\phi_{r_{\sigma(n)}},
\end{equation*}
where $\phi_k$ is either $\psi_k$ or $\psi_k^*$,
and $\sigma\in S_n$ is a permutation such that $r_{\sigma(1)}\leq\cdots\leq r_{\sigma(n)}$.
Denote by $\psi(\xi),\psi^*(\xi)$ the generating series of free fermions:
\be
\label{eq-gen-fermi}
\psi(\xi)= \sum_{s\in\bZ+\half} \psi_s \xi^{-s-\half},
\qquad\qquad
\psi^*(\xi)= \sum_{s\in\bZ+\half} \psi_s^* \xi^{-s-\half},
\ee
and then the generating series of the operators $\{\alpha_n\}$ is:
\be
\label{eq-gen-bos}
\alpha(\xi)= \sum_{n\in\bZ} \alpha_n \xi^{-n-1} = :\psi(\xi)\psi^*(\xi): .
\ee
The operators $\{\alpha_n\}_{n\in \bZ}$ are called the bosons,
and they satisfy:
\be
\label{eq-comm-boson}
[\alpha_m,\alpha_n]= m\delta_{m+n,0} \cdot \Id,
\ee
i.e.,
they generate a Heisenberg algebra.
One can also check that:
\be
\label{eq-comm-alphapsi}
[\alpha_n, \psi_r] = \psi_{n+r},
\qquad\qquad
[\alpha_n, \psi_r^*] = -\psi_{n+r}^*.
\ee

Let $\Lambda$ be the infinite-dimensional linear space of symmetric functions in some formal variables
$\bm x=(x_1,x_2,\cdots)$,
and let $w$ be a formal variable.
The space $\cB=\Lambda[[w,w^{-1}]]$ is called the bosonic Fock space.
The boson-fermion correspondence is a linear isomorphism $\Phi:\cF \to \cB$ of linear spaces
given by (see e.g. \cite[\S 5]{djm}):
\be
\Phi:\quad
|\bm a\rangle \in \cF^{(m)}
\quad\mapsto\quad
w^m\cdot \langle m |
e^{\sum_{n=1}^\infty \frac{p_n}{n} \alpha_n}
| \bm a \rangle,
\ee
where $p_n \in \Lambda$ is the Newton symmetric function of degree $n$.
In particular,
by restricting to $\cF^{(0)}$ one obtains an isomorphism:
\be
\label{eq-b-f-corresp}
\cF^{(0)}\to \Lambda,
\qquad
|\mu\rangle \mapsto s_\mu =
\langle 0 | e^{\sum_{n=1}^\infty \frac{p_n}{n} \alpha_n} | \mu \rangle,
\ee
where $|\mu\rangle$ is the basis vector \eqref{eq-basisF-mu}
and  $s_\mu$ is the Schur function indexed by $\mu$.
Moreover,
one has:
\be
\label{eq-boson}
\Phi (\alpha_n | \bm a\rangle)
=\begin{cases}
n\frac{\pd}{\pd p_n} \Phi(|\bm a\rangle), & n>0;\\
p_{-n}\cdot \Phi(|\bm a\rangle), & n<0,
\end{cases}
\ee
and:
\begin{equation*}
\Phi(\psi(\xi)|\bm a\rangle)=
\Psi(\xi) \Phi(|\bm a\rangle),
\qquad\quad
\Phi(\psi^*(\xi)|\bm a\rangle)=
\Psi^*(\xi) \Phi(|\bm a\rangle),
\end{equation*}
where $\Psi(\xi),\Psi^*(\xi)$ are the vertex operators:
\begin{equation*}
\begin{split}
&\Psi(\xi)=
\exp\bigg( \sum_{n=1}^\infty \frac{p_n}{n} \xi^n \bigg)
\exp\bigg( -\sum_{n=1}^\infty \xi^{-n}\frac{\pd}{\pd p_n} \bigg)
e^K \xi^{\alpha_0},\\
&\Psi^*(\xi)=
\exp\bigg( -\sum_{n=1}^\infty \frac{p_n}{n} \xi^n \bigg)
\exp\bigg( \sum_{n=1}^\infty \xi^{-n}\frac{\pd}{\pd p_n} \bigg)
e^{-K} \xi^{-\alpha_0},
\end{split}
\end{equation*}
and the actions of $e^K$ and $\xi^{\alpha_0}$ are defined by:
\begin{equation*}
(e^K f) (z,T)=z\cdot f(z,T),
\qquad\quad
(\xi^{\alpha_0} f)(z,T) = f(\xi z,T).
\end{equation*}

Now we recall Sato's theory for the KP hierarchy.
A tau-function of the KP hierarchy can be regarded as either a vector in the bosonic Fock space
or a vector in the fermionic Fock space,
satisfying the Hirota bilinear relation.
In Sato's theory,
the space of all formal power series tau-functions $\tau(\mathbf{t}) = \tau(t_1,t_2,t_3,\cdots)$
of the KP hierarchy
is a semi-infinite dimensional Grassmannian called the Sato Grassmannian,
which is the orbit of the trivial tau-function $\tau (\mathbf{t}) \equiv 1$
under the action of the infinite-dimensional Lie group $\widehat{GL}(\infty)$.
The Hirota bilinear relation for a group element $G$ is:
\be
\label{eq-Hirota-G}
[G\otimes G, \sum_{r\in \bZ+\half} \psi_r^* \otimes \psi_{-r}] = 0.
\ee

Let $\tau(\mathbf{t})$ be a tau-function of the KP hierarchy
with time variables $\mathbf{t} =(t_1,t_2,\cdots)$ satisfying the initial condition $\tau(\bm 0)\not= 0$,
then it lies on the big cell of the Sato Grassmannian.
There is a natural coordinate system $\{a_{n,m}\in \bC\}_{n,m\geq 0}$ on the big cell of the Sato Grassmannian
called the affine coordinates,
and a tau-function $\tau(\mathbf{t})$ satisfying $\tau(\bm 0)\not= 0$
can be completely described in terms of its affine coordinates $\{a_{n,m}\}_{n,m\geq 0}$
together with the constant $\tau(\bm 0)$,
see e.g. \cite{djm,zhou1, by} for details.
In the fermionic Fock space,
the tau-function corresponds to a Bogoliubov transform of the fermionic vacuum:
\be
\tau = c\cdot \exp\Big(
\sum_{n,m\geq 0} a_{n,m} \psi_{-m-\half} \psi_{-n-\half}^*
\Big) |0\rangle,
\ee
where $c = \tau(\bm 0)\not= 0$.
And in the bosonic Fock space,
the tau-function admits an expansion of Schur functions
whose coefficients are determinants of its affine coordinates:
\be
\tau(\mathbf{t}) = c\cdot
 \sum_{\mu\in \cP} (-1)^{n_1+\cdots+n_k}
\cdot \det \big(a_{n_i,m_j}\big)_{1\leq i,j\leq k} \cdot s_\mu(\mathbf{t}),
\ee
where $\mu = (m_1,\cdots,m_k|n_1,\cdots,n_k)$ is the Frobenius notation
of the partition $\mu$.
In particular,
$a_{n,m}$ is $(-1)^n$ times the coefficient of $s_{(m|n)} = s_{(m+1,1^n)}$.

\subsection{Some operators on the fermionic Fock space}

In this subsection we recall the properties of some operators on
the fermionic Fock space $\cF$.
We focus on the following three operators:
\be
\label{eq-def-C&Jopr}
\begin{split}
&C=\sum_{s\in \bZ+\half} :\psi_{-s} \psi_{s}^*:,\\
&L = \sum_{s\in \bZ+\half} s :\psi_{-s} \psi_{s}^*:
= \sum_{s>0} s \big( \psi_{-s}\psi_s^* + \psi_{-s}^*\psi_s \big),\\
&K = \sum_{s\in \bZ+\half} s^2
:\psi_{-s} \psi_{s}^*:
=\sum_{s>0} s^2 \big( \psi_{-s}\psi_s^* -\psi_{-s}^* \psi_{s} \big).
\end{split}
\ee
Here $C = \alpha_0$ is called the charge operator,
and $K$ is called the cut-and-join operator.
It is well-known that $\cF^{(n)}$ is the eigenspace of $C$
with respect to the eigenvalue $n\in \bZ$,
and for every basis vector $|\mu\rangle \in \cF^{(0)}$ one has:
\be
\label{eq-eigen-C&J}
\begin{split}
L |\mu\rangle = |\mu| \cdot |\mu\rangle, \qquad\qquad
K |\mu \rangle  =  \kappa_\mu |\mu\rangle,
\end{split}
\ee
where the eigenvalue $\kappa_\mu$ is given by \eqref{eq-def-kappamu}.

Now let $\mathbf{t} =(t_1,t_2,t_3,\cdots)$ be a sequence of formal variables,
and the vertex operators $\Gamma_\pm$ are defined to be the following operators on $\cF$:
\be
\label{eq-def-gammat}
\Gamma_\pm (\mathbf{t}) = \exp\Big( \sum_{n=1}^\infty t_n\alpha_{\pm n} \Big).
\ee
Let $z$ be a formal variable,
and let $\{z\}$ be the sequence $\{z\} = (z,\frac{z^2}{2},\frac{z^3}{3},\cdots)$.
For simplicity we will denote by $\Gamma_\pm (z)$ the following:
\be
\label{def-vert-Gamma}
\Gamma_\pm (z) = \Gamma_\pm ( \{z\} )
= \exp\Big( \sum_{n=1}^\infty \frac{z^n}{n} \alpha_{\pm n} \Big).
\ee
They satisfy the following commutation relation (see e.g. \cite[(A.11)]{Ok2}):
\be
\label{eq-commGamma}
\Gamma_+ (w) \Gamma_- (z) = \frac{1}{1-wz}
\Gamma_-(z) \Gamma_+(w).
\ee
The fermionic fields $\psi(z),\psi^*(z)$ can be recovered from $\Gamma_{\pm}$ by
(c.f. \cite[Ch.14]{kac}):
\be
\label{eq-psi-inGamma}
\begin{split}
\psi(z) = z^{C-1} R \Gamma_- (\{z\}) \Gamma_+ (-\{z^{-1}\}),\qquad
\psi^*(z) = R^{-1} z^{-C} \Gamma_- (-\{z\}) \Gamma_+ (\{z^{-1}\}),
\end{split}
\ee
where $C$ is the charge operator,
and $R$ is the shift operator on $\cF$ defined by:
\be
\label{eq-def-shiftop}
R (z^{a_1} \wedge z^{a_2} \wedge z^{a_3} \wedge \cdots ) =
z^{a_1-1} \wedge z^{a_2-1} \wedge z^{a_3-1} \wedge \cdots.
\ee

\section{Open Gromov--Witten Theory of Local Toric Calabi--Yau Threefolds}
\label{sec-TV}

In this section
we consider the open string amplitudes on a local toric Calabi--Yau threefold with one outer brane.
We first recall the topological vertex formalism for the Gromov--Witten theory
of toric Calabi--Yau threefolds \cite{adkmv, akmv, lllz, moop}.
For a local toric Calabi-Yau threefold $X$,
we use a fermionic representation of the topological vertex developed in \cite{wyz}
to represent the partition function of the open string amplitudes on $X$
as the trace of some operators on the fermionic Fock space.
We follow the notations in \cite{wyz}.

\subsection{The topological vertex formalism}

First we briefly recall the  topological vertex formalism
for the open Gromov--Witten theory of toric Calabi--Yau threefolds \cite{akmv, lllz,moop}.

The topological vertex $W_{\mu^1,\mu^2,\mu^3}(q)$ is the generating series of open Gromov--Witten invariants of $\bC^3$
with three special D-branes.
The notion of the topological vertex was first introduced by Aganagic {\em et al} \cite{akmv} using a string theoretical approach.
A mathematical theory for the topological vertex was developed in the setting of
the open GW theory of $\mathbb{C}^3$ by Li {\em et al} in \cite{lllz},
and in that work the localization technique was used in
the study of the local relative GW theory of toric Calabi--Yau threefolds.
The topological vertex constructed in \cite{lllz}
looks slightly different from the original version $W_{\mu^1,\mu^2,\mu^3}(q)$ proposed by physicists in \cite{akmv,adkmv},
and the equivalence of these two expressions was first proved in \cite{moop}  using GW/DT correspondence,
and then by a pure combinatorial approach in \cite{nt}.

It is known that the toric diagram of $\bC^3$ has a single trivalent vertex.
It serves as the basic building block.
For an arbitrary toric Calabi--Yau threefold $X$,
the toric diagram of $X$ is a trivalent planar graph which can be constructed by gluing such blocks,
and geometrically $X$ can be constructed accordingly by gluing some copies of $\bC^3$  together.
The open Gromov--Witten invariants of $X$ can be obtained by certain gluing procedure of the topological vertices
\cite{akmv}.
In such a procedure,
one needs to multiply the topological vertex by some extra factors along the edges
to get the framed topological vertex.
This is a sort of Feynman rules as in quantum field theory.

In the rest of this subsection,
we recall the combinatorial expressions for the topological vertex $W_{\mu^1,\mu^2,\mu^3}(q)$ given in \cite{akmv}.
Let $\mu^1$, $\mu^2$, $\mu^3$ be three partitions of integers,
and then the topological vertex $W_{\mu^1,\mu^2,\mu^3} (q)$ is defined by:
\begin{equation*}
W_{\mu^1,\mu^2,\mu^3}(q) = q^{\kappa_{\mu^2}/2 + \kappa_{\mu^3}/2}
\sum_{\rho^1,\rho^3} c_{\rho^1(\rho^3)^t}^{\mu^1(\mu^3)^t}\cdot
\frac{W_{(\mu^2)^t \rho^1}(q) W_{\mu^2 (\rho^3)^t}(q)}{W_{\mu^2}(q)},
\end{equation*}
where $c_{\rho^1(\rho^3)^t}^{\mu^1(\mu^3)^t} = \sum_\eta
c_{\eta \rho^1}^{\mu^1} c_{\eta (\rho^3)^t}^{(\mu^3)^t}$
and $c^\lambda_{\mu\nu}$ are the Littlewood-Richardson coefficients, and:
\begin{equation*}
\begin{split}
& W_{\mu}(q) =q^{\kappa_\mu /4}
\prod_{1\leq i<j \leq l(\mu)} \frac{[\mu_i-\mu_j+j-i]}{[j-i]}
\prod_{i=1}^{l(\mu)} \prod_{v=1}^{\mu_i} \frac{1}{[v-i+l(\mu)]},\\
& W_{\mu,\nu} (q) = q^{|\nu|/2} W_\mu (q) \cdot s_\nu (\cE_\mu(q,t)).
\end{split}
\end{equation*}
Here $s_\nu (\cE_\mu(q,t))$ is given by (see \cite[\S 5.2]{zhou5}):
\begin{equation*}
s_\nu (\cE_\mu(q,t)) = s_\nu (q^{\mu_1 -1},q^{\mu_2-2},\cdots).
\end{equation*}

In general,
one also need to consider the following framed topological vertex with framing $(\bm a)=(a_1,a_2,a_3) \in \bZ^3$
introduced in \cite{akmv}:
\be
W_{\mu^1,\mu^2,\mu^3}^{(\bm a)}(q)
= q^{a_1\kappa_{\mu^1}/2 + a_2\kappa_{\mu^2}/2 + a_3\kappa_{\mu^3}/2} \cdot
W_{\mu^1,\mu^2,\mu^3}(q),
\ee
or more explicitly,
\be
\label{eq-framedTV-def}
\begin{split}
W_{\mu^1,\mu^2,\mu^3}^{(\bm a)}(q)
=& (-1)^{|\mu^2|} \cdot
q^{a_1\kappa_{\mu^1}/2 + a_2\kappa_{\mu^2}/2 + (a_3+1)\kappa_{\mu^3}/2} \\
& \cdot s_{(\mu^2)^t} (q^{-\rho})
\sum_\eta s_{\mu^1 /\eta} (q^{(\mu^2)^t + \rho}) s_{(\mu^3)^t /\eta} (q^{\mu^2 + \rho}).
\end{split}
\ee
The framed topological vertex encodes the local relative Gromov--Witten invariants of $\bC^3$ with torus action specified by the framing $(a_1,a_2,a_3)$.
The framing plays an important role in the gluing of topological vertices.
The original topological vertex is recovered from $W_{\mu^1,\mu^2,\mu^3}^{(\bm a)}(q)$ by simply taking the framing $a_1=a_2=a_3=0$.

\subsection{Bogoliubov transform representation of the topological vertex}
In this subsection we recall the Bogoliubov transform representations of the topological vertex
conjectured by Aganagic {\em et al} \cite{adkmv} (the ADKMV Conjecture)
and proved by the authors \cite{wyz}.

The $N$-component fermionic Fock space is the tensor product of $N$ copies of the fermionic Fock space $\cF$,
i.e., $\cF_1\otimes \cF_2 \otimes \cdots \otimes \cF_N$.
One can define the $N$-component fermions $\{\psi_r^i,\psi_r^{i*}\}$ for $r\in \bZ+\half$ and $i=1,2,\cdots,N$ by introducing their actions on the $N$-component fermionic Fock space,
such that $\psi_r^i,\psi_r^{i*}$ act on the $i$-th component $\cF_i$
in the same way as the action of $\psi_r,\psi_r^{*}$ on $\cF$.
One thus can verify the following anti-commutation relations:
\be
\label{eqn:def multi fermions}
\begin{split}
	 [\psi_r^i,\psi_s^j]  =  [\psi_r^{i*},\psi_s^{j*}] =0,
	\qquad \qquad
	 [\psi_r^i,\psi_s^{j*}] = \delta_{r+s,0} \cdot \delta_{i,j} \cdot \Id.
\end{split}
\ee
for every $r,s\in\bZ+\half$ and $1\leq i,j\leq N$.

In \cite{akmv, adkmv},
Aganagic {\em et al} presented a conjectural fermionic representation
of the topological vertex $W_{\mu^1,\mu^2,\mu^3}(q)$ as a Bogoliubov transform of the three-component fermionic vacuum
in $\cF\otimes \cF\otimes \cF$,
which is called the ADKMV Conjecture.
They gave the explicit expressions for the coefficients of this Bogoliubov transform.
A straightforward consequence of their conjecture is the $3$-component KP integrability
of the topological vertex,
since it is well-known that a Bogoliubov transform in the $N$-component fermionic Fock space
corresponds to a tau-functions of the $N$-component KP hierarchy via boson-fermion correspondence \cite{kv}.
In \cite{dz1},
Deng and third-named author proposed a generalization of this conjecture
which takes the framing $\bm a = (a_1,a_2,a_3) \in \bZ^3$ into consideration.
This generalization is called the framed ADKMV conjecture.
More precisely,
the framed ADKMV Conjecture is the following representation of the framed topological vertex
as a vector in $3$-component fermionic Fock space:
\be
\label{eq-framedADKMV}
\begin{split}
	 W_{\mu^1,\mu^2,\mu^3}^{(\bm a)}(q)
	= \langle \mu^1,\mu^2,\mu^3|
	\exp\Big( \sum_{i,j =1,2,3} \sum_{m,n\geq 0} A_{mn}^{ij}(q; \bm a)
	\psi_{-m-\half}^{i} \psi_{-n-\half}^{j*} \Big)
	|0\rangle \otimes |0\rangle \otimes |0\rangle,
\end{split}
\ee
where the coefficient $A_{mn}^{ij}(q;\bm a)$ are given by the following combinatorial formulas:
\be
\label{eq-def-Aqa}
\begin{split}
	&A_{mn}^{ii}(q;\bm a) = (-1)^n q^{\frac{(2a_i+1) (m(m+1) - n(n+1))}{4}} \cdot
	\frac{1}{[m+n+1] \cdot [m]![n]!},\\
	&A_{mn}^{i(i+1)} (q;\bm a) = (-1)^n q^{\frac{(2a_i+1)m(m+1) - (2a_{i+1}+1) n(n+1)}{4} +\frac{1}{6}}
	\sum_{l=0}^{\min(m,n)} \frac{q^{\half (l+1)(m+n-l)}}{[m-l]![n-l]!},\\
	&A_{mn}^{i(i-1)} (q;\bm a) = (-1)^{n+1} q^{\frac{(2a_i+1)m(m+1) - (2a_{i-1}+1) n(n+1)}{4} -\frac{1}{6}}
	\sum_{l=0}^{\min(m,n)} \frac{q^{- \half  (l+1)(m+n-l)}}{[m-l]![n-l]!}.\\
\end{split}
\ee
We have used the conventions $A_{mn}^{34} = A_{m,n}^{31}$, $A_{mn}^{10}=A_{mn}^{13}$, and
\be
[m]! = \prod_{k=1}^m [k] = \prod_{k=1}^m (q^{k/2} - q^{-k/2}).
\ee
for an integer $m\geq 1$, and $[0]! =1$.
The original ADKMV Conjecture in \cite{adkmv} is the special case $a_1=a_2=a_3=0$ of \eqref{eq-framedADKMV}.

In \cite{dz1},
Deng and third-named author proved the one-legged case (i.e., $\mu^2 = \mu^3 =(\emptyset)$ and $a_2=a_3=0$)
and the two-legged case (i.e., $\mu^3 = (\emptyset)$ and $a_3=0$)
of the framed ADKMV Conjecture.
The general case has been proved by the authors in \cite{wyz} recently.

\subsection{Gluing of the topological vertices}
\label{sec:gluing prop}
In this subsection,
we review the gluing rule for the topological vertex,
see \cite{adkmv, dz2}.

As a consequence of the framed ADKMV conjecture,
the framed topological vertex in the fermionic vector
is a state in $\cF \otimes \cF \otimes \cF$ given by a Bogoliubov transform,
therefore,
it gives  a tau-function of the $3$-component KP hierarchy.
Now let $X$ be a toric CY threefold whose toric diagram contains $N$ external legs.
If the toric diagram of $X$ does not contain a loop,
i.e., when the toric diagram is a tree,
one can use the fermionic gluing procedure of topological vertex \cite{dz2}
to conclude that the generating series of open GW invariants of $X$ gives a tau-function
of the $N$-component KP hierarchy.

When the toric diagram of $X$ contains loops
(e.g. the total space of the line bundle $\cO(-3)$ over the complex projective plane $\bP^2$),
the situation  becomes more involved.
In this case the generating series of the open GW invariants of $X$
do not correspond to a tau-function of the $N$-component KP hierarchy
since it only involves the zero fermion number sector,
see \cite[\S 4.7; \S 5.13]{adkmv} for details.
In that work,
Aganagic--Dijkgraaf--Klemm--Mari\~{n}o--Vafa proposed to consider the fermionic number fluxes through the loops
in the toric diagram.
(This is similar to the integral over loop momenta in quantum field theory computations
for Feynman graphs with loops.)

\subsection{A new fermionic representation of the topological vertex}

The above approach is conceptually simple,  but is not easy to explicitly carry out in practice.
In this work we use a different fermionic representation of the topological vertex developed
in the proof of the ADKMV Conjecture  in \cite{wyz}.
A great advantage of this approach is  that it enables us to apply Wick's Theorem in computations.
In this subsection we recall this new representation.
For simplicity we only state the results for the case $a_1 =a_2 =a_3 = 0$.

Define the modified vertex operators $\tGa_\pm(z)$ by:
\be
\label{eq-def-tGamma}
\begin{split}
\tGa_+ (z) = f_1(z) \Gamma_+(z) R^{-1} z^C,\qquad\qquad
\tGa_- (z) = f_2(z) z^C R \Gamma_-(z),
\end{split}
\ee
where $\Gamma_\pm$ are the vertex operators \eqref{def-vert-Gamma}
and $R$ is the shift operator \eqref{eq-def-shiftop},
and the functions $f_1,f_2$ are given by:
\be
\label{eq-def-f1f2}
\begin{split}
f_1(q^k) = (-1)^{k-\half} \cdot q^{\frac{k^2}{2} - \frac{k}{2} -\frac{1}{16}},\qquad\qquad
f_2(q^k) = (-1)^{k+\half} \cdot q^{ - \frac{k}{2} -\frac{1}{16}}.
\end{split}
\ee
In this paper the variable $z$ in $f_i(z)$ is always of the form $q^k$ for some $k$.
The operators $\tGa_\pm (z)$
can be rewritten as follows
(see \cite[\S 4.1]{wyz}):
\be
\label{eq-tGa-reform}
\begin{split}
\tGa_+ (z) = f_1(z) \Gamma_-(z^{-1}) \psi^* (z^{-1}), \qquad\qquad
\tGa_- (z) = z\cdot f_2(z) \psi(z) \Gamma_+ (z^{-1}).
\end{split}
\ee
For  a partition $\mu = (m_1,\cdots,m_k|n_1,\cdots,n_k)$,
define the operator $\Psi_\mu (q)$ by:
\be
\label{eq-def-Psi}
\begin{split}
	\Psi_{\mu} (q) =& \prod_{j\geq 0}^{\longleftarrow} \Gamma_{-,\{j,\mu\}}(q^{-j-\half})
	\cdot \prod_{i\geq 0}^{\longrightarrow} \Gamma_{+,\{i,\mu\}}(q^{-i-\half})\\
	=& \cdots \Gamma_{-,\{2,\mu\}}(q^{-\frac{5}{2}}) \Gamma_{-,\{1,\mu\}}(q^{-\frac{3}{2}})
	\Gamma_{-,\{0,\mu\}}(q^{-\frac{1}{2}})
	\Gamma_{+,\{0,\mu\}}(q^{-\frac{1}{2}})
	\Gamma_{+,\{1,\mu\}}(q^{-\frac{3}{2}}) \cdots,
\end{split}
\ee
where
\begin{equation*}
	\begin{split}
		\Gamma_{-,\{j,\mu\}}(z) =
		\begin{cases}
			\Gamma_-(z), & \text{ if $j\notin \{m_1,\cdots,m_k\}$};\\
			\tGa_+(z^{-1}), & \text{ if $j\in \{m_1,\cdots,m_k\}$},
		\end{cases}
	\end{split}
\end{equation*}
and
\begin{equation*}
	\begin{split}
		\Gamma_{+,\{i,\mu\}}(z) =
		\begin{cases}
			\Gamma_+(z), & \text{ if $i\notin \{n_1,\cdots,n_k\}$};\\
			\tGa_-(z^{-1}), & \text{ if $i\in \{n_1,\cdots,n_k\}$.}
		\end{cases}
	\end{split}
\end{equation*}
One of the main results of \cite{wyz} is the following fermionic representation of the topological vertex $W_{\mu^1,\mu^2,\mu^3} (q)$
as an expectation value on the (one-component) fermionic Fock space $\cF$:
\begin{Theorem}
	[\cite{wyz}]
	Let $\mu^1,\mu^2,\mu^3$ be three partitions.
	Then:
	\be
	\label{eq-fermrep-wyz}
	W_{\mu^1,\mu^2,\mu^3} (q) =
	\langle \mu^1 |  \Psi_{\mu^2}(q) q^{-K /2} |(\mu^3)^t \rangle,
	\ee
	where $K$ is the cut-and-join operator \eqref{eq-def-C&Jopr}.
\end{Theorem}

Later we will see that
this fermionic representation of $W_{\mu^1,\mu^2,\mu^3} (q)$ plays an important role in the study
of the integrability of the open GW invariants of local toric CY threefolds.

\subsection{Open string amplitudes of local toric Calabi--Yau geometry}

Let $S$ be a nonsingular toric Fano surface,
and let $X$ be the total space of the canonical bundle on $S$.
Such a manifold $X$ is called a local toric Calabi--Yau threefold,
and the toric diagram associated with $X$ consists of
a polygon together with an external leg attached to each vertex of the polygon.
See Figure \ref{fig1} for the toric diagrams of $K_{\bP^2}$ and $K_{\bP^1\times \bP^1}$.
The topological vertex formalism developed in \cite{akmv}
gives an algorithm to compute the open and closed string amplitudes on $X$ using certain gluing procedures
similar to Feynman rules in quantum field theory.
The topological closed string theory of local toric Calabi--Yau geometries have been studied in \cite{ek, gz, pe, iq} and references therein.
For the computations of the closed string amplitudes,
since there are no outer D-branes and so the external legs
are all given  the empty partition $(\emptyset)$,
we only need to take summations over the partitions, one for each internal edge,
see e.g. \cite[(2.4), (2.5)]{ek}:
\begin{equation*}
\begin{split}
Z_{\text{closed}}^{(S)} =  \sum_{\mu^1,\cdots, \mu^M \in \cP} &
(-1)^{\sum_{i=1}^M \gamma_i |\mu^i|}
q^{\half \sum_{i=1}^M (\gamma_i+1) \kappa_{\mu^i}}
\cdot \prod_{i=1}^M Q_i^{|\mu^i|} \\
& \quad \cdot
W_{(\emptyset),\mu^{M,t},\mu^1} W_{(\emptyset),\mu^{1,t},\mu^2}
\cdots W_{(\emptyset),\mu^{M-1,t},\mu^M},
\end{split}
\end{equation*}
where $\gamma_i$ is the self-intersection number of the rational curve associated with the $i$-th edge,
and $Q_i= e^{-s_i}$ are the K\"ahler parameters.
Here $\mu^{i,t}$ denotes the transpose of $\mu^i$.

Now consider the open string amplitudes on the local toric CY threefold $X$
which allow nontrivial input data from the outer branes.
In general,
the gluing procedure of the topological vertex gives the following open string amplitudes:
\begin{equation}
\begin{split}
Z_{\lambda^1,\cdots,\lambda^M}^{(S)} (Q_i,\gamma_i) =  \sum_{\mu^1,\cdots, \mu^M \in \cP} &
(-1)^{\sum_{i=1}^M \gamma_i |\mu^i|}
q^{\half \sum_{i=1}^M (\gamma_i+1) \kappa_{\mu^i}}
\cdot \prod_{i=1}^M Q_i^{|\mu^i|} \\
& \cdot
W_{\mu^{M,t},\lambda^1,\mu^1} W_{\mu^{1,t},\lambda^2,\mu^2}
\cdots W_{\mu^{M-1,t},\lambda^M,\mu^M},
\end{split}
\end{equation}
where the partitions $\lambda^1,\cdots,\lambda^M\in \cP$ are the input on external legs in the toric diagram.
For simplicity we will focus on the special case of one outer brane,
i.e., the case $\lambda^2 = \cdots = \lambda^M = (\emptyset)$,
and denote:
\be
\label{eq-def-Zlambda}
Z_\lambda (Q_i,\gamma_i) = Z_{\lambda,(\emptyset),\cdots,(\emptyset)}^{(S)} (Q_i,\gamma_i).
\ee
 To illustrate the idea,  let $S= \bP^2$,
 then $\gamma_1=\gamma_2=\gamma_3 =1$,
 and
 \begin{equation*}
\begin{split}
Z_{\lambda}^{(\bP^2)} (Q_i) =  \sum_{\mu^1,\mu^2, \mu^3 \in \cP}
(-1)^{\sum_{i=1}^3  |\mu^i|}
q^{ \sum_{i=1}^3  \kappa_{\mu^i}}
\prod_{i=1}^3 Q^{|\mu^i|}
 \cdot
W_{\mu^{3,t},\lambda,\mu^1} W_{\mu^{1,t},(\emptyset),\mu^2}
W_{\mu^{2,t},(\emptyset),\mu^3},
\end{split}
\end{equation*}
as can be visually depicted in the following picture:
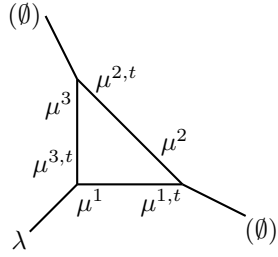
\begin{equation*}
\begin{tikzpicture}[scale=0.7]
\draw [thick] (0,0) -- (0,2);
\draw [thick] (0,0) -- (2,0);
\draw [thick] (2,0) -- (0,2);
\draw [thick] (0,0) -- (-0.9,-0.9);
\draw [thick] (0,2) -- (-0.6,3.2);
\draw [thick] (2,0) -- (3.2,-0.6);
\node [align=center,align=center] at (-1.1,-1.05) {$\lambda$};
\node [align=center,align=center] at (0.25,-0.35) {$\mu^1$};
\node [align=center,align=center] at (1.6,-0.35) {$\mu^{1,t}$};
\node [align=center,align=center] at (3.5,-0.9) {$(\emptyset)$};
\node [align=center,align=center] at (-1,3.2) {$(\emptyset)$};
\node [align=center,align=center] at (1.85,0.8) {$\mu^2$};
\node [align=center,align=center] at (0.75,2) {$\mu^{2,t}$};
\node [align=center,align=center] at (-0.35,1.5) {$\mu^3$};
\node [align=center,align=center] at (-0.45,0.4) {$\mu^{3,t}$};
\end{tikzpicture}
\end{equation*}

\subsection{Partition function from fermionic point of view}

The summations over partitions can be understood by assigning a bosonic Fock space $\Lambda$ to each internal edge.
If we cyclically label the vertices by $i=1, \cdots, n$,
and denote by $\Lambda_{(i)}$ the copy of bosonic Fock space associated with the edge from $i$ to $i+1$,
then each $(-1)^{  \gamma_i |\mu^i|}
q^{\half   (\gamma_i+1) \kappa_{\mu^i}}
 \  Q_i^{|\mu^i|}
W_{\mu^{i-1}, \lambda^i, \mu^i}$ can be understood as
the matrix coefficient of a linear operator $\Phi_{\lambda^i}$ form $\Lambda_{(i-1)}$ to $\Lambda_{(i)}$,
where $\Lambda_{(0)} = \Lambda_{(M)}$.
Then the partition function can be understood as a trace on $\Lambda_{(0)}$
of the product $\Phi_{\lambda^1} \cdots \Phi_{\lambda^M}$.

Now we rewrite the above open string amplitudes $Z_\lambda(Q_i,\gamma_i)$
as a trace in the fermionic picture.
Using the fermionic representation \eqref{eq-fermrep-wyz}
we can rewrite $Z_\lambda(Q_i,\gamma_i)$ as follows:
\be
\begin{split}
Z_\lambda   ( Q_i,\gamma_i) =  &
\sum_{ \mu^1,\cdots, \mu^M \in \cP}
(-1)^{\sum_{i=1}^M \gamma_i |\mu^i|}
q^{\half \sum_{i=1}^M (\gamma_i+1) \kappa_{\mu^i}}
\cdot \Big( \prod_{i=1}^M Q_i^{|\mu^i|} \Big) \\
& \qquad\qquad
\cdot
W_{\mu^{M,t},\lambda,\mu^1} W_{\mu^{1,t},(\emptyset),\mu^2}
\cdots W_{\mu^{M-1,t},(\emptyset),\mu^M} \\
=&
\sum_{ \mu^1,\cdots, \mu^M \in \cP}
(-1)^{\sum_{i=1}^M \gamma_i |\mu^i|}
q^{\half \sum_{i=1}^M (\gamma_i+1) \kappa_{\mu^i}}
\Big( \prod_{i=1}^M Q_i^{|\mu^i|} \Big)
\big\langle \mu^{M,t} \big| \Psi_\lambda(q) q^{-\half K} \big|\mu^{1,t} \big\rangle\\
& \qquad\qquad
\cdot
\big\langle \mu^{1,t} \big| \Psi_{(\emptyset)}(q) q^{-\half K} \big|\mu^{2,t} \big\rangle \cdots
\big\langle \mu^{M-1,t} \big| \Psi_{(\emptyset)}(q) q^{-\half K} \big|\mu^{M,t} \big\rangle.
\end{split}
\ee
Notice that $\kappa_{\mu^t} = -\kappa_\mu$ for every partition $\mu$,
and thus by \eqref{eq-eigen-C&J} we have:
\begin{equation*}
(-1)^{ \gamma_i |\mu^i|}\cdot
q^{\half (\gamma_i+1) \kappa_{\mu^i}} \cdot Q_i^{|\mu^i|}
 \big|\mu^{i,t} \big\rangle
= q^{-\half (\gamma_i+1)K} (-1)^{\gamma_i\cdot L} Q_i^L \big|\mu^{i,t} \big\rangle,
\end{equation*}
and then
\begin{equation*}
\begin{split}
Z_\lambda  ( Q_i,\gamma_i)
=&
\sum_{ \mu^1,\cdots, \mu^M \in \cP}
\big\langle \mu^{M,t} \big| \Psi_\lambda(q) q^{-\half (\gamma_1+2) K}
(-1)^{\gamma_1\cdot L} Q_1^L \big|\mu^{1,t} \big\rangle \\
& \qquad\qquad
 \cdot \big\langle \mu^{1,t} \big| \Psi_{(\emptyset)}(q) q^{-\half (\gamma_2+2) K}
(-1)^{\gamma_2\cdot L} Q_2^L \big|\mu^{2,t} \big\rangle \cdots \\
& \qquad\qquad
\cdot
\big\langle \mu^{M-1,t} \big| \Psi_{(\emptyset)}(q) q^{-\half (\gamma_M+2) K}
(-1)^{\gamma_M\cdot L} Q_M^L  \big|\mu^{M,t} \big\rangle.
\end{split}
\end{equation*}
Then we use the property
$\sum_{\mu\in \cP} |\mu\rangle \langle\mu| = \Id_{\cF^{(0)}}$ on $\cF^{(0)}$, and get:
\be
\label{eq-fermpf-Zlambda}
\begin{split}
&Z_\lambda   ( Q_i,\gamma_i)
= \sum_{\mu\in \cP}
\Big\langle \mu \Big|
\Big( \Psi_\lambda(q) q^{-\half (\gamma_1+2) K}
(-1)^{\gamma_1\cdot L} Q_1^L \Big)  \\
&\quad \cdot
\Big( \Psi_{(\emptyset)}(q) q^{-\half (\gamma_2+2) K}
(-1)^{\gamma_2\cdot L} Q_2^L \Big)
 \cdots
\Big( \Psi_{(\emptyset)}(q) q^{-\half (\gamma_M+2) K}
(-1)^{\gamma_M\cdot L} Q_M^L \Big)  \Big|\mu \Big\rangle.
\end{split}
\ee
Recall that $\{|\mu\rangle\}$ is a basis for the fermionic Fock space $\cF^{(0)}$
of charge zero,
and now it is clear that $Z_\lambda(Q_i,\gamma_i)$ is a trace of a product of operators
on $\cF^{(0)}$.
It is natural to consider the trace of such an operator on the whole fermionic Fock space $\cF$,
and this is exactly the way we construct the total partition function
which contains the contributions of all fermion number fluxes.
We will give the details of this construction in the next section.

\section{Total Partition Function for Local Toric Calabi--Yau Threefolds}
\label{sec:def Ztotal}

In this section we introduce the total partition function
for a local Calabi-Yau threefold $X$
which contains contributions of all fermion number fluxes
by extending the trace \eqref{eq-fermpf-Zlambda}
to a trace on the whole fermionic Fock space $\cF$.
We also give an explicit formula representing the total partition function
in terms of the open string amplitudes of $X$ in the bosonic picture.

\subsection{Construction of the contribution of fermion number flux $N$}

In this subsection we propose a mathematical construction of the contribution of fermion number flux $N$
for every integer $N\in \bZ$,
and define the total partition function for a general local toric CY threefold.

Let $X$ be a (formal) local CY threefold
with K\"ahler parameters $Q_i$ and framing factors $\gamma_i$ for $i=1,\cdots,M$,
and denote by
\be
Z(\mathbf{t};Q_i,\gamma_i) = \sum_{\lambda \in \cP}
Z_\lambda (Q_i,\gamma_i) \cdot s_\lambda(\mathbf{t})
\ee
the partition function of the open string amplitudes on $X$,
where $s_\lambda(\mathbf{t})$ is the Schur polynomial in $\mathbf{t}$ indexed by the partition $\lambda$.
Let $R$ be the shift operator \eqref{eq-def-shiftop} on the fermionic Fock space,
then $\{R^{-N}|\mu \rangle\}_{N\in \bZ,\mu\in \cP}$ is a basis of the fermionic Fock space $\cF$.
We define the contribution of fermion number flux $N$ and type $\lambda$ by
conjugating $R^{N}$ to the operator in the right-hand side of \eqref{eq-fermpf-Zlambda}:
\be
\label{eq-def-ZlambdaN}
\begin{split}
&Z_\lambda^{(N)}   ( Q_i,\gamma_i ; \Xi)
= \sum_{\mu\in \cP}
\Big\langle \mu \Big|R^N \Xi^C
\Big( \Psi_\lambda(q) q^{-\half (\gamma_1+2) K}
(-1)^{\gamma_1\cdot L} Q_1^L \Big) \\
&\quad \cdot
\Big( \Psi_{(\emptyset)}(q) q^{-\half (\gamma_2+2) K}
(-1)^{\gamma_2\cdot L} Q_2^L \Big)
 \cdots
\Big( \Psi_{(\emptyset)}(q) q^{-\half (\gamma_M+2) K}
(-1)^{\gamma_M\cdot L} Q_M^L \Big) R^{-N} \Big|\mu \Big\rangle,
\end{split}
\ee
where $\Xi$ is a formal variable which tracks the fermion number flux $N$.
In particular, for $N=0$ one has:
\be
Z^{(0)}_\lambda(Q_i,\gamma_i ;\Xi) = Z_\lambda(Q_i,\gamma_i).
\ee
As in the derivation of \eqref{eq-fermpf-Zlambda},
if we use the property
$\sum_{n\in \bZ} \sum_{\mu\in \cP} R^n|\mu\rangle \langle\mu|R^{-n} = \Id_{\cF}$ on $\cF$ instead,
the right-hand side of \eqref{eq-def-ZlambdaN} can be rewritten in the form of a Feynman rule as in last subsection.
The difference is now we associate a copy of fermionic Fock space $\cF_{(i)}$ with the edge from vertex $i$ to vertex $i+1$,
and regard \eqref{eq-def-ZlambdaN} as a trace on $\cF_{(0)}$.

\begin{Definition}
Let $X$ be a (formal) local CY threefold with K\"ahler parameters $Q_1,\cdots,Q_M$
and framing factors $\gamma_1,\cdots,\gamma_M$,
and let $\mathbf{t} =(t_1,t_2,\cdots)$ be a sequence of formal variables.
Define the total partition function for $X$ by:
\be
\label{eq-Ztotaldef}
\begin{split}
&Z^{\text{total}} (\mathbf{t}; Q_i,\gamma_i; \Xi)
= \sum_{N\in \bZ} \sum_{\lambda \in \cP}
Z_\lambda^{(N)}   ( Q_i,\gamma_i ; \Xi) \cdot s_\lambda(\mathbf{t}).
\end{split}
\ee
\end{Definition}

Then our main theorem in this section is the following:
\begin{Theorem}
\label{thm-Ztotal-reformulate}
The total partition function $Z^{\text{total}}(\mathbf{t};Q_i,\gamma_i;\Xi)$ for the local toric CY threefold $X$
is of the following form:
\be
Z^{\text{total}} (\mathbf{t};Q_i,\gamma_i;\Xi) =
\sum_{N\in \bZ} Z^{(N)}(\mathbf{t}; Q_i,\gamma_i)\cdot \Xi^{-N}.
\ee
The $N$-th partition function $Z^{(N)}$ is given by:
\be
\label{eq-def-zN}
 Z^{(N)} (\mathbf{t};Q_i,\gamma_i)
= q^{-NL_0} \cdot Q^{\frac{N^2}{2}} (-1)^{\gamma \frac{ N^2}{2}}
q^{(\gamma +2M) \frac{N(4N^2 -1)}{24}}
Z(\mathbf{t}; q^{(\gamma_i+2)N}Q_i,\gamma_i),
\ee
where $L_0$ is the following degree operator acting on formal power series in $\mathbf{t}$:
\be
L_0 = \sum_{n=1}^\infty nt_n \frac{\pd}{\pd t_n},
\ee
and we use the notations:
\be
Q = \prod_{i=1}^M Q_i,
\qquad\qquad
\gamma = \sum_{i=1}^M \gamma_i.
\ee
In particular,
for $N=0$ we simply have
\be
Z^{(0)} (\mathbf{t};Q_i,\gamma_i) = Z(\mathbf{t};Q_i,\gamma_i).
\ee
\end{Theorem}

The proof of this theorem
will be given in \S \ref{sec-pf-Ztotal}.
Recall that $s_\lambda (\mathbf{t})$ is a weighted homogeneous polynomial of degree $|\lambda|$
in $t_1,\cdots, t_{|\lambda|}$ if we assign $\deg (t_i) = i$,
and then:
\begin{equation*}
L_0 \big( s_\lambda (\mathbf{t}) \big) = |\lambda| \cdot s_\lambda (\mathbf{t})
\end{equation*}
for every $\lambda$.
Therefore the $N$-th partition function can be rewritten as:
\be
\label{eq-defZN-2}
 Z^{(N)} (\mathbf{t};Q_i,\gamma_i)
= \sum_{\lambda\in \cP}
q^{-N| \lambda |}  Q^{\frac{N^2}{2}} (-1)^{\gamma \frac{ N^2}{2}}
q^{(\gamma +2M) \frac{N(4N^2 -1)}{24}} \cdot
Z_\lambda(q^{(\gamma_i+2)N}Q_i,\gamma_i)s_\lambda(\mathbf{t}).
\ee

\begin{Remark}
The constant term of the $N$-th partition function $Z^{(N)}$ is:
\be
Z^{(N)} (\bm 0; Q_i,\gamma_i) =
Q^{\frac{N^2}{2}} \cdot (-1)^{\gamma \frac{N^2}{2}} \cdot q^{(\gamma+2M)\frac{N(4N^2 -1)}{24}}
\cdot Z (\bm 0; q^{(\gamma_i+2)N}Q_i,\gamma_i).
\ee
In particular,
in the local $\bP^2$ case one has $M=3$ and $\gamma_i=1$,
and then by taking $N=\pm 1$ one obtains:
\be
\label{eq-P2const}
Z^{(\pm 1)} (\bm 0; Q_i,\bm 1) =
(Q_1Q_2Q_3)^{1/2} (-1)^{3/2} q^{\pm 9/8}
\cdot Z (\bm 0; q^{\pm 3}\cdot Q_i, \bm 1).
\ee
Aganagic {\em et al} proposed a similar formula for local $\bP^2$
by using a fermionic propagator in the gluing procedure,
see \cite[(5.54)]{adkmv}.
They checked their formula up to order $2$.
The structure of our formula \eqref{eq-P2const} matches with theirs.
\end{Remark}

\begin{Remark}
Even though our results are confined to the case of one-loop toric diagrams,
we hope they may shed some light on the properties of the partition function
when the mirror curve has genus greater than one.
It will be interesting to
compare our construction of the total partition function
with the non-perturbative constructions given in \cite{abdks, be, gjz}
from the viewpoints of the topological recursion \cite{eo} and background independent matrix integrals.
\end{Remark}

\subsection{Some commutation relations}

We recall some commutation relations between operators on the fermionic Fock space $\cF$ in this subsection.
These commutation relations will be useful in the rest of this paper.

First we recall the following relations between the shift operator $R$
and the operators $C,L,K$ defined by \eqref{eq-def-C&Jopr},
which can be found in \cite[Appendix A]{Ok2} and \cite[\S 2.7]{Ok1}:
\begin{Lemma}
\label{lem-comm-RN}
Let $R$ be the shift operator \eqref{eq-def-shiftop}, then:
\be
\label{eq-comm-RN}
\begin{split}
& R^N C R^{-N} = C - N,\\
& R^N L R^{-N} = L-NC + \frac{N^2}{2},\\
& R^N K R^{-N} = K -2NL +N^2 C - \frac{N(4N^2 -1)}{12},
\end{split}
\ee
for every integer $N$.
\end{Lemma}

Moreover, we have the following:
\begin{Lemma}
\label{lem-comm-RN-2}
Let $\Gamma_\pm(\mathbf{t})$ be the vertex operators defined by \eqref{eq-def-gammat},
and let $\Psi_\mu(q)$ be the operators defined by \eqref{eq-def-Psi}.
Then for every $N\in\bZ$, we have:
\be
\label{eq-comm-RN-2}
\begin{split}
& R^N \Gamma_\pm (\mathbf{t}) R^{-N} = \Gamma_\pm (\mathbf{t}), \\
& R^N \Psi_\mu (q) R^{-N} = q^{-N |\mu|}\cdot \Psi_\mu (q).
\end{split}
\ee
\end{Lemma}
\begin{proof}
The first equality follows from the fact $R \alpha_n R^{-1} = \alpha_n$
for every $n$.
Now we prove the second one.
From the definition \eqref{eq-def-tGamma} of the operators $\tGa_\pm (z)$ we have:
\begin{equation*}
\begin{split}
R^N \tGa_+ (z) R^{-N}
=& f_1(z) R^N \Gamma_+(z) R^{-1} z^C R^{-N} \\
=& f_1(z) \Gamma_+(z) R^{N-1} z^C R^{-N} \\
=& f_1(z) \Gamma_+(z) R^{-1} z^{C-N}\\
=& z^{-N} \cdot \tGa_+ (z),
\end{split}
\end{equation*}
where we have used Lemma \ref{lem-comm-RN} and the first equality in \eqref{eq-comm-RN-2}.
Similarly,
one has:
\begin{equation*}
R^N \tGa_- (z) R^{-N} = z^{-N} \cdot \tGa_- (z).
\end{equation*}
Then for a partition $\mu = (m_1,\cdots,m_k | n_1,\cdots,n_k)$,
from \eqref{eq-def-Psi} we get:
\begin{equation*}
\begin{split}
R^N \Psi_\mu (q) R^{-N}
=& q^{\sum_{i=1}^k (-N)(m_i+\half) + \sum_{i=1}^k (-N)(n_i+\half)}  \cdot \Psi_\mu (q)
= q^{-N\cdot |\mu|} \cdot \Psi_\mu (q),
\end{split}
\end{equation*}
which completes the proof.
\end{proof}

\begin{Lemma}
Let $a$ be a formal parameter.
For every partition $\mu$,
we have:
\be
\label{eq-comm-CLPsi}
\begin{split}
& a^C \Psi_\mu (q) a^{-C} = \Psi_\mu (q).
\end{split}
\ee
\end{Lemma}
\begin{proof}
The conclusion follows from
the fact that $\Psi_\mu (q)$ is of charge zero.
\end{proof}

\begin{Lemma}
\label{lem-comm-LGamma}
We have:
\be
\label{eq-comm-LGamma}
a^L \Gamma_\pm (z) a^{-L} = \Gamma_\pm (a^{\mp 1} \cdot z).
\ee
\end{Lemma}
\begin{proof}
It suffices to check the special case $a=e$,
since the general case can be obtained easily by replacing $L$ by $\log a \cdot L$.
We need the Baker--Campbell--Hausdorff formula:
\be
\label{eq-BCH}
e^XYe^{-X} = Y + [X,Y] +\frac{1}{2!}[X,[X,Y]] + \frac{1}{3!}[X,[X,[X,Y]]] +\cdots.
\ee
First we compute the commutator of $L$ and $\alpha_n$.
By the definition of $L$ we have:
\begin{equation*}
[L,\alpha_n] =
\sum_{s>0} s\big( [\psi_{-s}\psi_s^* ,\alpha_n]
+ [\psi_{-s}^*\psi_s ,\alpha_n] \big).
\end{equation*}
Then by \eqref{eq-comm-alphapsi},
we have:
\begin{equation*}
\begin{split}
& [\psi_{-s}\psi_s^* ,\alpha_n]
=
[\psi_{-s} ,\alpha_n] \psi_s^*
+
\psi_{-s} [\psi_s^* ,\alpha_n]
= \psi_{-s}\psi_{n+s}^* - \psi_{n-s}\psi_s^*, \\
& [\psi_{-s}^*\psi_s ,\alpha_n]
=
[\psi_{-s}^* ,\alpha_n] \psi_s
+
\psi_{-s}^* [\psi_s ,\alpha_n]
= -\psi_{-s}^*\psi_{n+s} + \psi_{n-s}^* \psi_s,
\end{split}
\end{equation*}
and thus:
\begin{equation*}
[L,\alpha_n] = \sum_{s>0}s \big( \psi_{-s}\psi_{n+s}^* - \psi_{n-s}\psi_s^* \big)
+\sum_{s>0} s \big( -\psi_{-s}^*\psi_{n+s} + \psi_{n-s}^* \psi_s \big) =
-n\cdot \alpha_n.
\end{equation*}
Then by \eqref{eq-BCH} we obtain:
\begin{equation*}
e^L \alpha_n e^{-L} = e^{-n} \cdot \alpha_n,
\end{equation*}
for every $n\in \bZ$.
Notice that $\{\alpha_n\}_{n>0}$ commute with each others, and then:
\begin{equation*}
e^L \Gamma_+(z) e^{-L}
= e^{-L} \exp\Big( \sum_{n=1}^\infty \frac{z^n}{n}\alpha_n \Big) e^L
= \exp\Big( \sum_{n=1}^\infty \frac{z^n e^{-n}}{n}\alpha_n \Big)
= \Gamma_+ (e^{-1}z).
\end{equation*}
Similarly,
one has:
\begin{equation*}
e^L \Gamma_- (z) e^{-L}
= e^{-L} \exp\Big( \sum_{n=1}^\infty \frac{z^n}{n}\alpha_{-n} \Big) e^L
= \exp\Big( \sum_{n=1}^\infty \frac{z^n e^{n}}{n}\alpha_{-n} \Big)
= \Gamma_- (ez),
\end{equation*}
and thus the conclusion holds.
\end{proof}

\begin{Lemma}
\label{lem-comm-RL}
Let $a$ be a formal parameter.
We have:
\be
\begin{split}
 a^L R^{-1} a^{-L} = a^{\half -C} R^{-1}, \qquad\qquad
 a^L R a^{-L} = R a^{C-\half }.
\end{split}
\ee
\end{Lemma}
\begin{proof}
From the second equality in \eqref{eq-comm-RN} we know that
\begin{equation*}
[L,R^{-1}] = (\half -C) R^{-1}.
\end{equation*}
Here $L$ is of charge $0$ and thus $[L,C] = 0$.
Then by the BCH formula \eqref{eq-BCH} one obtains:
\begin{equation*}
e^L R^{-1} e^{-L} = e^{\half -C} R^{-1}.
\end{equation*}
The second equality can be proved similarly.
\end{proof}

The following corollary will be useful in the computations in \S \ref{sec-eg}.
\begin{Corollary}
\label{cor-comm-LPsi}
For a partition $\mu = (m_1,\cdots,m_k | n_1,\cdots,n_k)$,
one has:
\be
\label{eq-comm-LPsi}
a^L \Psi_\mu(q)a^{-L} =
\prod_{j\geq 0}^{\longleftarrow} \Gamma_{-,\{j,\mu\}}(q^{-j-\half};a)
\cdot \prod_{i\geq 0}^{\longrightarrow} \Gamma_{+,\{i,\mu\}}(q^{-i-\half};a),
\ee
where:
\begin{equation*}
\begin{split}
\Gamma_{-,\{j,\mu\}}(z;a) =
\begin{cases}
\Gamma_-(az), & \text{ if $j\notin \{m_1,\cdots,m_k\}$};\\
f_1(z^{-1})
\Gamma_+ (a^{-1}z^{-1})  a^{\half -C} R^{-1}  z^{-C},
 & \text{ if $j\in \{m_1,\cdots,m_k\}$},
\end{cases}
\end{split}
\end{equation*}
and
\begin{equation*}
\begin{split}
\Gamma_{+,\{i,\mu\}}(z;a) =
\begin{cases}
\Gamma_+(a^{-1}z), & \text{ if $i\notin \{n_1,\cdots,n_k\}$};\\
f_2(z^{-1})  z^{-C} R a^{C-\half}  \Gamma_-(az^{-1}),
& \text{ if $i\in \{n_1,\cdots,n_k\}$,}
\end{cases}
\end{split}
\end{equation*}
where $f_1,f_2$ are the functions \eqref{eq-def-f1f2}.
In particular,
for the empty partition $(\emptyset)$ one has:
\be
\label{eq-comm-LPsiempty}
a^L \Psi_{(\emptyset)}(q)a^{-L} =
\prod_{j\geq 0}^{\longleftarrow} \Gamma_{-}(aq^{-j-\half})
\cdot \prod_{i\geq 0}^{\longrightarrow} \Gamma_{+}(a^{-1} q^{-i-\half}).
\ee
\end{Corollary}
\begin{proof}
First we use Lemma \ref{lem-comm-LGamma} to compute the conjugations of the operators $\tGa_\pm(z)$
(see \eqref{eq-def-tGamma}) by $a^L$:
\begin{equation*}
\begin{split}
a^L \tGa_+ (z) a^{-L} =& f_1(z) \cdot
a^L \Gamma_+(z) a^{-L} \cdot a^L R^{-1} a^{-L}  \cdot a^L z^C a^{-L} \\
=& f_1(z) \cdot \Gamma_+ (a^{-1}z)
\cdot a^L R^{-1} a^{-L}  \cdot a^L z^C a^{-L}.
\end{split}
\end{equation*}
Here $a^L z^C a^{-L} = z^C$ since $[L,C] = 0$,
and $a^L R^{-1} a^{-L}$ has been computed in Lemma \ref{lem-comm-RL}.
Therefore one has:
\begin{equation*}
a^L \tGa_+ (z) a^{-L} = f_1(z) \cdot
\Gamma_+ (a^{-1}z) \cdot a^{\half -C} R^{-1} \cdot z^C.
\end{equation*}
Similarly one can prove that:
\begin{equation*}
a^L \tGa_- (z) a^{-L} = f_2(z) \cdot z^C \cdot R a^{C-\half} \cdot \Gamma_-(az).
\end{equation*}
Then we apply this and Lemma \ref{lem-comm-LGamma} to \eqref{eq-def-Psi},
and prove the conclusion.
\end{proof}

\subsection{Proof of Theorem \ref{thm-Ztotal-reformulate}}
\label{sec-pf-Ztotal}

We start with the computation
of the right-hand side of \eqref{eq-def-ZlambdaN}.
By Lemma \ref{lem-comm-RN} and Lemma \ref{lem-comm-RN-2} we have:
\begin{equation*}
\begin{split}
& R^N\cdot \Xi^C \cdot \Psi_\lambda (q) q^{-\half (\gamma_1+2)K} (-1)^{\gamma_1 \cdot L} Q_1^L \cdot R^{-N} \\
=& R^N \Xi^C R^{-N} R^N \Psi_\lambda (q) R^{-N}
R^N q^{-\half (\gamma_1+2)K} R^{-N}
R^N (-1)^{\gamma_1 \cdot L} R^{-N}
R^N Q_1^L R^{-N} \\
=& \Xi^{C-N}
q^{-N|\mu|} \Psi_\lambda (q)
q^{-\half (\gamma_1+2) (K-2NL+ N^2 C -\frac{N(4N^2 -1)}{12})}
  (-1)^{\gamma_1(L-NC + \frac{N^2}{2})}
 Q_1^{L-NC + \frac{N^2}{2}}\\
=& \Xi^{-N} \cdot
q^{ -N|\lambda| -\half (\gamma_1+2) (N^2 C -\frac{N(4N^2 -1)}{12})} \cdot
(-1)^{\gamma_1(-NC + \frac{N^2}{2})} \cdot
Q_1^{-NC + \frac{N^2}{2}} \\
& \cdot \Xi^C \Psi_\lambda(q)
q^{-\half (\gamma_1+2) K}
(-1)^{\gamma_1\cdot L}
( q^{(\gamma_1 +2)N} Q_1)^L,
\end{split}
\end{equation*}
where we have used the fact
\be
[C,L]=[C,K]=[L,K]=0
\ee
which can be checked directly from the definitions of these operators.
Similarly
for $i=2,3,\cdots, M$, we have:
\begin{equation*}
\begin{split}
& R^N \cdot \Psi_{(\emptyset)}(q) q^{-\half (\gamma_i+2) K}
(-1)^{\gamma_i\cdot L} Q_i^L \cdot R^{-N} \\
=& \Psi_{(\emptyset)} (q)\cdot
q^{-\half (\gamma_i+2) (K-2NL+ N^2 C -\frac{N(4N^2 -1)}{12})} \cdot
(-1)^{\gamma_i(L-NC + \frac{N^2}{2})} \cdot
Q_i^{L-NC + \frac{N^2}{2}} \\
=&
q^{ -\half (\gamma_i+2) (N^2 C -\frac{N(4N^2 -1)}{12})} \cdot
(-1)^{\gamma_i(-NC + \frac{N^2}{2})} \cdot
Q_i^{-NC + \frac{N^2}{2}} \\
& \cdot  \Psi_{(\emptyset)}(q) \cdot
q^{-\half (\gamma_i+2) K} \cdot
(-1)^{\gamma_i\cdot L} \cdot
( q^{(\gamma_i +2)N} Q_i)^L.
\end{split}
\end{equation*}
And then the total partition function $Z^{\text{total}} (\mathbf{t}; Q_i,\gamma_i; \Xi)$
defined by \eqref{eq-Ztotaldef} becomes:
\begin{equation*}
\begin{split}
Z^{\text{total}} &(\mathbf{t}; Q_i,\gamma_i; \Xi) =
 \sum_{N\in\bZ} \sum_{\lambda,\mu \in \cP}
s_\lambda(\mathbf{t} )\cdot \Xi^{-N} q^{-N|\lambda|}
\Big( \prod_{i=1}^M q^{ \frac{1}{24} (\gamma_i+2) N(4N^2 -1)}
 (-1)^{\half \gamma_i N^2}  Q_i^{\half N^2} \Big)   \\
& \cdot \Big\langle \mu \Big|
\Big( q^{-\half (\gamma_1+2) N^2 C} (-1)^{-\gamma_1 NC} Q_1^{-NC}
\Xi^C \Psi_\lambda(q)
q^{-\half (\gamma_1+2) K}
(-1)^{\gamma_1 L}
( q^{(\gamma_1 +2)N} Q_1)^L \Big)\\
& \cdot\Big(
q^{-\half (\gamma_2+2) N^2 C} (-1)^{-\gamma_2 NC} Q_2^{-NC}
 \Psi_{(\emptyset)}(q)
q^{-\half (\gamma_2+2) K}
(-1)^{\gamma_2 L}
( q^{(\gamma_2 +2)N} Q_1)^L
\Big) \cdots \\
& \cdot \Big(
q^{-\half (\gamma_M+2) N^2 C} (-1)^{-\gamma_M NC} Q_M^{-NC}
 \Psi_{(\emptyset)}(q)
q^{-\half (\gamma_M+2) K}
(-1)^{\gamma_M L}
( q^{(\gamma_M +2)N} Q_1)^L
\Big) \Big| \mu \Big\rangle. \\
\end{split}
\end{equation*}
Notice that the vectors $\langle \mu|$, $| \mu\rangle$
and the operators $L$, $K$, $\Psi_\lambda(q)$ are of charge $0$,
thus:
\begin{equation*}
\begin{split}
Z^{\text{total}} (\mathbf{t}; Q_i,\gamma_i; \Xi) =
 \sum_{N\in\bZ} \sum_{\lambda,\mu \in \cP} &
s_\lambda(\mathbf{t} )\cdot \Xi^{-N} q^{-N|\lambda|}
\Big( \prod_{i=1}^M q^{ \frac{1}{24} (\gamma_i+2) N(4N^2 -1)}
 (-1)^{\half \gamma_i N^2}  Q_i^{\half N^2} \Big)  \\
& \cdot  \Big\langle \mu \Big|
\Big( \Psi_\lambda(q)
q^{-\half (\gamma_1+2) K}
(-1)^{\gamma_1 L}
( q^{(\gamma_1 +2)N} Q_1)^L \Big)  \\
& \qquad \cdot \Big(
 \Psi_{(\emptyset)}(q)
q^{-\half (\gamma_2+2) K}
(-1)^{\gamma_2 L}
( q^{(\gamma_2 +2)N} Q_1)^L
\Big) \cdots \\
& \qquad \cdot \Big(
 \Psi_{(\emptyset)}(q)
q^{-\half (\gamma_M+2) K}
(-1)^{\gamma_M L}
( q^{(\gamma_M +2)N} Q_1)^L
\Big) \Big| \mu \Big\rangle. \\
\end{split}
\end{equation*}
Then by \eqref{eq-fermpf-Zlambda} and \eqref{eq-defZN-2} we get:
\begin{equation*}
\begin{split}
& Z^{\text{total}} (\mathbf{t}; Q_i,\gamma_i; \Xi) \\
= &
 \sum_{N\in\bZ} \sum_{\lambda \in \cP}
s_\lambda(\mathbf{t} )\cdot \Xi^{-N} q^{-N|\lambda|}
\Big( \prod_{i=1}^M q^{ \frac{1}{24} (\gamma_i+2) N(4N^2 -1)}
 (-1)^{\half \gamma_i N^2}  Q_i^{\half N^2} \Big)
 \cdot  Z_\lambda   ( Q_i,\gamma_i) \\
 =&
\sum_{N\in \bZ} Z^{(N)}(\mathbf{t}; Q_i,\gamma_i)\cdot \Xi^{-N}.
\end{split}
\end{equation*}
This completes the proof.

\section{KP Integrability of the Total Partition Function}
\label{sec:KP integrability}

In this section we derive a determinantal formula for the total partition function $Z^{\text{total}}$
of a local toric Calabi--Yau threefold,
and then show that $Z^{\text{total}}$ is a tau-function of the KP hierarchy.
We also give a formula for its affine coordinates on the Sato Grassmannian.

\subsection{A determinantal formula}

Now in this subsection we prove a determinantal formula
which will be useful in the proof of the KP integrability of the total partition function.

\begin{Lemma}
\label{lem-det}
Let $\varphi_1, \cdots, \varphi_k$
be linear (infinite) summations of the free fermions $\{\psi_r\}_{r\in \bZ +\half}$,
and let $\varphi_1^*, \cdots, \varphi_k^* $
be linear (infinite) summations of $\{\psi_r^*\}_{r\in \bZ +\half}$.
Suppose that $G_1^*,\cdots,G_k^*$ and $G_1,\cdots,G_k$ are some fermionic operators satisfying the
Hirota bilinear relation \eqref{eq-Hirota-G},
and assume that:
\begin{equation*}
\sum_{N\in \bZ}\sum_{\mu\in \cP} \langle\mu| R^N \cdot
 G_1^*  G_2^* \cdots  G_k^*
 G_k  \cdots  G_2 G_1
\cdot R^{-N} |\mu\rangle \not= 0.
\end{equation*}
Then:
\be
\label{eq-lem-det}
\begin{split}
& \frac{\sum_{N\in \bZ}\sum_{\mu\in \cP} \langle\mu| R^N \cdot
\varphi_1^* G_1^* \varphi_2^* G_2^* \cdots \varphi_k^* G_k^*
 \varphi_k G_k  \cdots \varphi_2 G_2 \varphi_1 G_1
\cdot R^{-N} |\mu\rangle}
{\sum_{N\in \bZ}\sum_{\mu\in \cP} \langle\mu| R^N \cdot
 G_1^*  G_2^* \cdots  G_k^*
 G_k  \cdots  G_2 G_1
\cdot R^{-N} |\mu\rangle} \\
=& \det(A_{ij})_{1\leq i,j\leq k},
\end{split}
\ee
where
\begin{equation*}
\begin{split}
A_{ij} =
 \frac{\sum_{N\in \bZ}\sum_{\mu\in \cP} \langle\mu| R^N
G_1^*   \cdots G_{i-1}^* \varphi_i^* G_i^*\cdots G_k^*
G_k
  \cdots  G_{j+1} \varphi_j G_j \cdots G_1
 R^{-N} |\mu\rangle}
{\sum_{N\in \bZ}\sum_{\mu\in \cP} \langle\mu| R^N
 G_1^* G_2^*  \cdots  G_k^*
 G_k  \cdots  G_2 G_1
 R^{-N} |\mu\rangle}.
\end{split}
\end{equation*}
\end{Lemma}

\begin{proof}
Consider the following function in $\{\varphi_i,\varphi_i^*\}_{1\leq i\leq k}$ and $\{G_i^*, G_i\}_{1\leq i\leq k}$:
\be
\label{eq-defB-phiG}
\begin{split}
B(\varphi_i,\varphi_i^*,G_i^*,G_i) = &
\sum_{r\in \bZ+\half} \sum_{N,M\in \bZ} \sum_{\mu,\nu\in \cP}
\langle \mu | R^N \psi_r \varphi_1^* G_1^*\cdots G_k^* G_k \cdots G_1 R^{-N}|\mu\rangle \\
& \quad \cdot
\langle \lambda | R^M \psi_{-r}^* G_1^* \varphi_2^* G_2^* \cdots \varphi_k^* G_k^*
\varphi_k G_k \cdots \varphi_1 G_1 R^{-M} |\lambda\rangle.
\end{split}
\ee
Suppose $\varphi_1^* = \sum_r c_r \psi_r^*$ for some $c_r \in \bC$,
then by the anti-commutation relation \eqref{eq-anticomm-psi}
one has $[\psi_r , \varphi_1^* ]_+ = c_{-r}$.
Now replace $\psi_r \varphi_1^*$ by $c_{-r}-\psi_r \varphi_1^*$ in \eqref{eq-defB-phiG},
and then we may express $B(\varphi_i,\varphi_i^*,G_i^*,G_i)$ as follows:
\be
\label{eq-BphiG-1}
\begin{split}
B(\varphi_i,\varphi_i^*,G_i^*,G_i) = \text{(I)} - \text{(II)},
\end{split}
\ee
where the two terms (I) and (II) are:
\be
\label{eq-BphiG-2}
\begin{split}
\text{(I)} = &
 \sum_{N,M\in \bZ} \sum_{\mu,\nu\in \cP}
\langle \mu | R^N \cdot G_1^*\cdots G_k^* G_k \cdots G_1 \cdot R^{-N}|\mu\rangle \\
& \quad \cdot
\langle \lambda | R^M \cdot \varphi_1^* G_1^* \varphi_2^* G_2^* \cdots \varphi_k^* G_k^*
\varphi_k G_k \cdots \varphi_1 G_1 \cdot R^{-M} |\lambda\rangle,
\end{split}
\ee
and
\be
\label{eq-BphiG-2}
\begin{split}
\text{(II)} = &
\sum_{r\in \bZ+\half} \sum_{N,M\in \bZ} \sum_{\mu,\nu\in \cP}
\langle \mu | R^N \cdot \varphi_1^* \psi_r \cdot G_1^*\cdots G_k^* G_k \cdots G_1 \cdot R^{-N}|\mu\rangle \\
& \quad \cdot
\langle \lambda | R^M \cdot \psi_{-r}^* G_1^* \varphi_2^* G_2^* \cdots \varphi_k^* G_k^*
\varphi_k G_k \cdots \varphi_1 G_1 \cdot R^{-M} |\lambda\rangle
\end{split}
\ee
respectively.
Here $G_1^*$ satisfies the bilinear relation \eqref{eq-Hirota-G}:
\begin{equation*}
\sum_{r\in \bZ+\half } \psi_r G_1^* \otimes \psi_{-r}^* G_1^*
= \sum_{r\in \bZ+\half }  G_1^*\psi_r \otimes  G_1^*\psi_{-r}^*,
\end{equation*}
and then:
\begin{equation*}
\begin{split}
\text{(II)} = &
\sum_{r\in \bZ+\half} \sum_{N,M\in \bZ} \sum_{\mu,\nu\in \cP}
\langle \mu | R^N \cdot \varphi_1^*  G_1^* \psi_r
\cdot G_2^* \cdots G_k^* G_k \cdots G_1 \cdot R^{-N}|\mu\rangle \\
& \quad \cdot
\langle \lambda | R^M \cdot G_1^* \psi_{-r}^* \cdot \varphi_2^* G_2^* \cdots \varphi_k^* G_k^*
\varphi_k G_k \cdots \varphi_1 G_1 \cdot R^{-M} |\lambda\rangle \\
=& -\sum_{r\in \bZ+\half} \sum_{N,M\in \bZ} \sum_{\mu,\nu\in \cP}
\langle \mu | R^N \cdot \varphi_1^*  G_1^* \psi_r
\cdot G_2^* \cdots G_k^* G_k \cdots G_1 \cdot R^{-N}|\mu\rangle \\
& \quad \cdot
\langle \lambda | R^M \cdot G_1^* \varphi_2^* \cdot \psi_{-r}^* G_2^* \cdots \varphi_k^* G_k^*
\varphi_k G_k \cdots \varphi_1 G_1 \cdot R^{-M} |\lambda\rangle,
\end{split}
\end{equation*}
where we have also used $[\varphi_2^* , \psi_{-r}^*]_+ = 0$.
Similarly
we can commute $\sum_r \psi_r\otimes \psi_{-r}^*$ with $G_i^*\otimes G_i^*$
for $i=2,3,\cdots,k$ successively,
and in this way we obtain:
\begin{equation*}
\begin{split}
\text{(II)}
=& (-1)^{k-1}
\sum_{r\in \bZ+\half} \sum_{N,M\in \bZ} \sum_{\mu,\nu\in \cP}
\langle \mu | R^N  \varphi_1^*  G_1^* G_2^* \cdots G_k^* \cdot\psi_r\cdot
 G_k \cdots G_1  R^{-N}|\mu\rangle \\
& \quad \cdot
\langle \lambda | R^M \cdot G_1^*  \varphi_2^* G_2^* \cdots \varphi_k^* G_k^* \cdot \psi_{-r}^*\cdot
\varphi_k G_k \cdots \varphi_1 G_1 \cdot R^{-M} |\lambda\rangle.
\end{split}
\end{equation*}
Now assume $\varphi_k = \sum_{r}\tilde{c}_r \psi_r$,
and then $[\psi_{-r}^*,\varphi_k]_+ = \tilde  c_r$.
Therefore
\begin{equation*}
\begin{split}
\text{(II)}
=\text{(III)} -\text{(IV)}
\end{split}
\end{equation*}
where
\begin{equation*}
\begin{split}
\text{(III)}=& (-1)^{k-1}
 \sum_{N,M\in \bZ} \sum_{\mu,\nu\in \cP}
\langle \mu | R^N  \varphi_1^*  G_1^* G_2^* \cdots G_k^* \cdot\varphi_k \cdot
 G_k \cdots G_1  R^{-N}|\mu\rangle \\
& \quad \cdot
\langle \lambda | R^M \cdot G_1^*  \varphi_2^* G_2^* \cdots \varphi_k^* G_k^* \cdot
 G_k \varphi_{k-1}G_{k-1} \cdots \varphi_1 G_1 \cdot R^{-M} |\lambda\rangle, \\
\text{(IV)} = &  (-1)^{k-1}
\sum_{r\in \bZ+\half} \sum_{N,M\in \bZ} \sum_{\mu,\nu\in \cP}
\langle \mu | R^N  \varphi_1^*  G_1^* G_2^* \cdots G_k^* \cdot\psi_r\cdot
 G_k \cdots G_1  R^{-N}|\mu\rangle \\
& \quad \cdot
\langle \lambda | R^M \cdot G_1^*  \varphi_2^* G_2^* \cdots \varphi_k^* G_k^* \cdot
\varphi_k  \psi_{-r}^* G_k \varphi_{k-1}G_{k-1} \cdots \varphi_1 G_1 \cdot R^{-M} |\lambda\rangle.
\end{split}
\end{equation*}
Now we can apply the bilinear relation \eqref{eq-Hirota-G}
to (IV) to commute $\psi_r, \psi_{-r}^*$ with $G_k$,
and then deal with the product $\psi_{-r}^* \varphi_{k-1}$ in the result using the same method.
Repeat this procedure,
and finally we obtain
\begin{equation*}
\begin{split}
\text{(II)} = \text{(V)} - \text{(VI)}
\end{split}
\end{equation*}
where
\begin{equation*}
\begin{split}
&\text{(V)} =
\sum_{i=1}^k (-1)^{i-1} \sum_{N,M\in \bZ} \sum_{\mu,\nu\in \cP}
\langle \mu | R^N  \varphi_1^* \cdot G_1^* G_2^* \cdots G_k^*
 G_k \cdots G_{i+1} \cdot \varphi_i \cdot G_i \cdots G_1  R^{-N}|\mu\rangle \\
& \qquad \cdot
\langle \lambda | R^M  G_1^* \cdot \varphi_2^* G_2^* \cdots \varphi_k^* G_k^* \cdot
 \varphi_k G_k  \cdots \varphi_{i+1} G_{i+1} \cdot G_i \cdot
 \varphi_{i-1} G_{i-1} \cdots \varphi_1 G_1 R^{-M} |\lambda\rangle,\\
& \text{(VI)} =
\sum_{r\in \bZ+\half} \sum_{N,M\in \bZ} \sum_{\mu,\nu\in \cP}
\langle \mu | R^N  \varphi_1^*  G_1^* G_2^* \cdots G_k^* \cdot
 G_k G_{k-1} \cdots G_1 \cdot \psi_r R^{-N}|\mu\rangle \\
& \qquad \cdot
\langle \lambda | R^M \cdot G_1^*  \varphi_2^* G_2^* \cdots \varphi_k^* G_k^* \cdot
\varphi_k  G_k \varphi_{k-1}G_{k-1} \cdots \varphi_1 G_1 \cdot \psi_{-r}^* R^{-M} |\lambda\rangle.
\end{split}
\end{equation*}
Notice that the linear functional $\sum_{N\in \bZ} \sum_{\mu\in \bP}\langle \mu|R^N (-) R^{-N} |\mu \rangle$
is actually taking the trace on $\cF$,
and thus
we can move the operators $\psi_r$ and $\psi_{-r}^*$ in (VI)
to the far left side and rewrite (VI) as:
\begin{equation*}
\begin{split}
\text{(VI)} =&  \sum_{r\in \bZ+\half} \sum_{N,M\in \bZ} \sum_{\mu,\nu\in \cP}
\langle \mu | R^N \cdot \psi_r\cdot \varphi_1^*  G_1^* G_2^* \cdots G_k^* \cdot
 G_k G_{k-1} \cdots G_1 \cdot  R^{-N}|\mu\rangle \\
&  \cdot
\langle \lambda | R^M \cdot \psi_{-r}^* \cdot G_1^*  \varphi_2^* G_2^* \cdots \varphi_k^* G_k^* \cdot
\varphi_k  G_k \varphi_{k-1}G_{k-1} \cdots \varphi_1 G_1 \cdot  R^{-M} |\lambda\rangle,
\end{split}
\end{equation*}
which is exactly $B(\varphi_i,\varphi_i^*,G_i^*,G_i)$, see \eqref{eq-defB-phiG}.
Recall $B(\varphi_i,\varphi_i^*,G_i^*,G_i) = \text{(I)} - \text{(II)}$,
and thus we conclude that $\text{(I)} = \text{(V)}$, i.e.,
\begin{equation*}
\begin{split}
& \sum_{N,M\in \bZ} \sum_{\mu,\nu\in \cP}
\langle \mu | R^N \cdot G_1^*\cdots G_k^* \cdot G_k \cdots G_1 \cdot R^{-N}|\mu\rangle \\
& \quad \cdot
\langle \lambda | R^M \cdot \varphi_1^* G_1^* \varphi_2^* G_2^* \cdots \varphi_k^* G_k^* \cdot
\varphi_k G_k \cdots \varphi_1 G_1 \cdot R^{-M} |\lambda\rangle \\
=&
\sum_{i=1}^k (-1)^{i-1} \sum_{N,M\in \bZ} \sum_{\mu,\nu\in \cP}
\langle \mu | R^N  \varphi_1^* \cdot G_1^* G_2^* \cdots G_k^*
 G_k \cdots G_{i+1} \cdot \varphi_i \cdot G_i \cdots G_1  R^{-N}|\mu\rangle \\
&  \cdot
\langle \lambda | R^M  G_1^* \cdot \varphi_2^* G_2^* \cdots \varphi_k^* G_k^* \cdot
\varphi_k G_k  \cdots \varphi_{i+1}G_{i+1} \cdot G_i \cdot
 \varphi_{i-1} G_{i-1} \cdots \varphi_1 G_1 R^{-M} |\lambda\rangle.
\end{split}
\end{equation*}
This matches with the Laplace expansion theorem for determinants,
and thus the determinantal formula \eqref{eq-lem-det} holds by induction on $k$.
\end{proof}

\begin{Remark}
	The determinantal formula in Lemma \ref{lem-det} is a generalization
of  the so-called Wick's lemma at finite temperature in literature
(see \cite{gau} and   \cite[Lemma B.1]{bb}).
\end{Remark}

\subsection{KP integrability and affine coordinates of $Z^{\text{total}}$}
\label{sec-KPaffine}

Now we are able to prove the KP integrability of
the total partition function $Z^{\text{total}}$.
Recall that $Z^{\text{total}}$ can be represented as:
\begin{equation*}
\begin{split}
&Z^{\text{total}} (\mathbf{t}; Q_i,\gamma_i; \Xi)
= \sum_{N\in \bZ} \sum_{\lambda,\mu \in \cP}  s_\lambda(\mathbf{t})\cdot
\Big\langle \mu \Big|
R^N \Xi^C
\Big( \Psi_\lambda (q) q^{-\half (\gamma_1+2)K} (-1)^{\gamma_1  L} Q_1^L \Big)\cdot  \\
& \quad
\Big( \Psi_{(\emptyset)}(q) q^{-\half (\gamma_2+2) K}
(-1)^{\gamma_2 L} Q_2^L \Big) \cdots
\Big( \Psi_{(\emptyset)}(q) q^{-\half (\gamma_M+2) K}
(-1)^{\gamma_M L} Q_M^L \Big)
R^{-N} \Big| \mu \Big\rangle,
\end{split}
\end{equation*}
where $\Psi_\mu(z)$ is a product of operators of the form $\Gamma_\pm(z)$ or $\tGa_\pm (z)$,
see  \eqref{eq-def-Psi}.
If we plug \eqref{eq-tGa-reform} into \eqref{eq-def-Psi},
we may see that $\Psi_{\mu}(q)$ is a product of operators of the form $\psi(z),\psi^*(z)$
and $\Gamma_\pm (z)$
(multiplied by some functions $f(z)$).
From the commutation relations \eqref{eq-comm-alphapsi}
and the Baker--Campbell--Hausdorff formula \eqref{eq-BCH}
we known that the conjugations
\begin{equation*}
\Gamma_\pm(z)^{-1} \psi_r \Gamma_\pm(z),
\qquad\qquad
\Gamma_\pm(z)^{-1} \psi_r^* \Gamma_\pm(z)
\end{equation*}
of the free fermions $\psi_r,\psi_r^*$
are linear summations of $\{\psi_s\}_{s\in \bZ+\half}$ and $\{\psi_s^*\}_{s\in \bZ+\half}$ respectively.
Moreover,
it is well-known that the operators $C,L,K$ defined by \eqref{eq-def-C&Jopr}
are elements of $\widehat{gl(\infty)}$
(see e.g. \cite{Ok1})
and thus the operators $\Xi^C$, $q^K$, $(-1)^L$, $Q_i^L$ in the above fermionic representation of $Z^{\text{total}}$
all satisfy the bilinear relation \eqref{eq-Hirota-G}.
And here $C$ is the charge operator which commutes with all $\Gamma_\pm(z)$.
Therefore we can apply Lemma \ref{lem-det} to the above fermionic representation of $Z^{\text{total}}$.
The result is:

\begin{Theorem} \label{thm:Determinantal}
The total partition function $Z^{\text{total}} (\mathbf{t}; Q_i,\gamma_i; \Xi)$
is given by the following Schur polynomial expansion:
\be
Z^{\text{total}} (\mathbf{t}; Q_i,\gamma_i; \Xi)
= c_0 \cdot \sum_{\lambda \in \cP}
\det (A_{m_i,n_j}^{\text{total}})_{1\leq i,j \leq k} \cdot s_\lambda(\mathbf{t}),
\ee
where $\lambda = (m_1,\cdots,m_k| n_1,\cdots,n_k)$ is the Frobenius notation for $\lambda$,
and
\be
\label{eq-def-Amn}
\begin{split}
& A_{m,n}^{\text{total}} = \frac{1}{c_0} \cdot \sum_{N\in \bZ} \sum_{\mu\in \cP}
\Big\langle \mu \Big|
R^N \Xi^C
\Big( \Psi_{(m|n)} (q) q^{-\half (\gamma_1+2)K} (-1)^{\gamma_1  L} Q_1^L \Big)\cdot  \\
& \quad
\Big( \Psi_{(\emptyset)}(q) q^{-\half (\gamma_2+2) K}
(-1)^{\gamma_2 L} Q_2^L \Big) \cdots
\Big( \Psi_{(\emptyset)}(q) q^{-\half (\gamma_M+2) K}
(-1)^{\gamma_M L} Q_M^L \Big)
R^{-N} \Big| \mu \Big\rangle,
\end{split}
\ee
and $c_0$ is the normalization constant:
\be
\begin{split}
 c_0 =  \sum_{N\in \bZ} \sum_{\mu\in \cP} &
\Big\langle \mu \Big|
R^N \Xi^C
\Big( \Psi_{(\emptyset)} (q) q^{-\half (\gamma_1+2)K} (-1)^{\gamma_1  L} Q_1^L \Big)  \\
& \cdots
\Big( \Psi_{(\emptyset)}(q) q^{-\half (\gamma_M+2) K}
(-1)^{\gamma_M L} Q_M^L \Big)
R^{-N} \Big| \mu \Big\rangle.
\end{split}
\ee
\end{Theorem}

\begin{Remark}
Here we regard the operators $\Psi_{(\emptyset)} (q) q^{-\half (\gamma_i+2)K} (-1)^{\gamma_i  L} Q_i^L$
as some formal series in $q$ and $Q_i$,
and in this sense one has $c_0 \not= 0$.
\end{Remark}

Now we combine this Schur polynomial expansion for $Z^{\text{total}}(\mathbf{t}; Q_i,\gamma_i ; \Xi)$
with the theory of KP hierarchy (see \S \ref{pre-KPtau}),
then we easily see that:
\begin{Corollary}
$Z^{\text{total}} (\mathbf{t}; Q_i,\gamma_i; \Xi)$ is a tau-function of the KP hierarchy
with KP-time variables $\mathbf{t} = (t_1,t_2,\cdots)$.
Moreover,
the affine coordinates of this tau-function is:
\be
a_{n,m}^{\text{total}} = (-1)^n \cdot A_{m,n}^{\text{total}},
\qquad n,m\geq 0,
\ee
where $\{A_{m,n}^{\text{total}}\}_{m,n\geq 0}$ are given by \eqref{eq-def-Amn}.
\end{Corollary}

In \cite{zhou1},
the third-named author has derived a formula to compute the connected $n$-point functions
of a tau-function of the KP hierarchy using its affine coordinates on the Sato Grassmannian.
Assume $\tau(\mathbf{t})$ is a tau-function of the KP hierarchy satisfying $\tau(\bm 0) \not=0$,
and let $\{a_{n,m}\}_{n,m\geq 0}$ be its affine coordinates.
Denote by
\begin{equation*}
A(\xi,\eta)= \sum_{m,n\geq 0}
a_{n,m} \xi^{-n-1} \eta^{-m-1}
\end{equation*}
the generating series of the affine coordinates,
then for $n\geq 1$ one has (see \cite[\S 5]{zhou1}):
\be
\begin{split}
&\sum_{j_1,\cdots,j_n\geq 1}
\frac{\pd^n \log\tau (\mathbf{t})}{\pd t_{j_1}\cdots \pd t_{j_n}}\bigg|_{\mathbf{t}=0}
\cdot z_1^{-j_1-1} \cdots z_n^{-j_n-1}\\
=& (-1)^{n-1} \sum_{n\text{-cycles }\sigma}
\prod_{i=1}^{n} \widehat A (z_{\sigma^{i}(1)}, z_{\sigma^{i+1}(1)})
-\frac{\delta_{n,2}}{(z_1-z_2)^2},
\end{split}
\ee
where
\begin{equation*}
\widehat A (z_i,z_j)=\begin{cases}
i_{z_i,z_j}\frac{1}{z_i-z_j} + A(z_i,z_j), & i<j;\\
A(z_i,z_i), & i=j;\\
i_{z_j,z_i}\frac{1}{z_i-z_j} + A(z_i,z_j), & i>j,
\end{cases}
\end{equation*}
and we use the notation
$i_{\xi,\eta}\frac{1}{\xi+\eta}= \sum_{k\geq 0} (-1)^k \xi^{-1-k}\eta^k$.

If we want to compute the connected $n$-point functions for $Z^{\text{total}} (\mathbf{t}; Q_i,\gamma_i; \Xi)$
(which gives the open Gromov--Witten invariants of local toric Calabi--Yau threefolds as examples)
using Zhou's formula,
we will need some simpler explicit formulas for the affine coordinates $\{a_{n,m}^{\text{total}}\}$
or their generating series.
However,
this seems to be hard for a general local toric CY threefold
and we do not have a general closed formula yet.
In the rest of this paper,
we will consider some special examples
and carry out some explicit calculations.

\section{Example: the $(-2,\dots, -2)$-Model }
\label{sec-eg}

Notice that for arbitrarily chosen $M\in \bN^+$ and parameters $\gamma_1,\cdots,\gamma_M \in \bC$,
we have shown that the total partition function  is a
(formal power series) tau-function of the KP hierarchy.
Though here for general $(\gamma_1,\cdots,\gamma_M )$ the tau-function  $Z^{\text{total}} (\mathbf{t}; Q_i,\gamma_i; \Xi)$
may not come from the open GW theory of a concrete toric Calabi--Yau threefold,
the tau-function $Z^{\text{total}} (\mathbf{t}; Q_i,\gamma_i; \Xi)$ is  always well-defined and we can study its properties.
In this section,
we will consider the special case $\gamma_1 = \cdots =\gamma_M = -2$.
Denote the tau-function in this case by $Z^{\text{total}} (\mathbf{t}; Q_i,-2; \Xi)$.
We will give explicit formulas for its constant term
and affine coordinates,
which, together with the KP integrability,
completely determine this tau-function.
For convenience,
we will always assume that $|Q_i|<1, \forall i$ from now on.

The exceptional divisor of the minimal resolution of  an isolated elliptic Kulikov surface singularity of type $I_{M-1}$
is a cycle of $(-2)$-rational curves \cite[p. 294]{nem}.
The geometry behind the $(-2, \dots, -2)$-model might be related to such a resolution multiplied by $\bC$.

\subsection{Constant term of $Z^{\text{total}} (\textbf{t}; Q_i,-2; \Xi)$}

In this subsection we compute the constant term
of the tau-function $Z^{\text{total}} (\textbf{t}; Q_i,-2; \Xi)$.
Recall that by \eqref{eq-def-ZlambdaN} we have:
\be
\label{eq-gamma-2-ferm}
\begin{split}
Z^{\text{total}} (\mathbf{t}; Q_i,-2; \Xi)
=& \sum_{N\in \bZ} \sum_{\lambda,\mu \in \cP}  s_\lambda(\mathbf{t})\cdot
\Big\langle \mu \Big|
R^N \Xi^C
\Big( \Psi_\lambda (q)  (-1)^{-2 L} Q_1^L \Big)  \\
& \cdot
\Big( \Psi_{(\emptyset)}(q)
(-1)^{-2 L} Q_2^L \Big) \cdots
\Big( \Psi_{(\emptyset)}(q)
(-1)^{-2 L} Q_M^L \Big)
R^{-N} \Big| \mu \Big\rangle.
\end{split}
\ee
Thus the constant term of this tau-function is:
\be
\label{eqn:Z0 as vev}
\begin{split}
&Z^{\text{total}} (\bm 0; Q_i,-2; \Xi) \\
=& \sum_{N\in \bZ} \sum_{\mu \in \cP}
\Big\langle \mu \Big|
R^N \Xi^C
\Big( \Psi_{(\emptyset)} (q)  (-1)^{-2 L} Q_1^L \Big) \cdots
\Big( \Psi_{(\emptyset)}(q)
(-1)^{-2 L} Q_M^L \Big)
R^{-N} \Big| \mu \Big\rangle .
\end{split}
\ee
Now let $\Theta_3 (t;q)$ be the theta-function:
\be
\Theta_3 (t;q) = \sum_{n\in \bZ} q^{n^2}  t^n.
\ee
and denote by $\cM(z;q)$ the following MacMahon function:
\be
\cM(z;q) = \prod_{n=1}^\infty (1-z q^{-n})^n.
\ee
Then we have the following:
\begin{Theorem}
\label{prop-Z-2-const}
The constant term $Z^{\text{total}} (\bm 0; Q_i,-2; \Xi)$ of the above tau-function is:
\be
\label{eq-Z-2-const}
\begin{split}
Z^{\text{total}} (\bm 0; Q_i,-2; \Xi)= &
\Theta_3 \big( (-1)^M \cdot \Xi^{-1}; Q^{1/2} \big) \cdot
\prod_{1\leq k <l \leq M}  \frac{1}{\cM (\prod_{i=k}^{l-1}Q_i ; q)} \\
& \cdot
\prod_{i=1}^\infty \frac{1}{1-Q^i} \cdot
\prod_{k,l= 1}^M \prod_{j=0}^\infty
\frac{1}{\cM(Q^j \cdot \prod_{a=1}^{k-1} Q_a \cdot \prod_{b=l}^M Q_i ;q)},
\end{split}
\ee
where $Q = Q_1Q_2\cdots Q_M$.
\end{Theorem}
\begin{proof}
By the commutation relations in Lemma \ref{lem-comm-RN} and Lemma \ref{lem-comm-RN-2} we have:
\begin{equation*}
\begin{split}
& R^N \Xi^C R^{-N} = \Xi^{C-N}, \\
& R^N (-1)^{-2L} Q_i^L R^{-N} = (-1)^{-2L+2NC-N^2} Q_i^{L-NC+N^2/2}, \\
& R^N \Psi_{(\emptyset)}(q) R^{-N} = \Psi_{(\emptyset)}(q),
\end{split}
\end{equation*}
and thus the second line of equation \eqref{eqn:Z0 as vev} is equal to
\begin{equation*}
\begin{split}
& \sum_{N\in \bZ} \sum_{\mu \in \cP}
\Big\langle \mu \Big|
R^N \Xi^C
\Big( \Psi_{(\emptyset)} (q)  (-1)^{-2 L} Q_1^L \Big) \cdots
\Big( \Psi_{(\emptyset)}(q)
(-1)^{-2 L} Q_M^L \Big)
R^{-N} \Big| \mu \Big\rangle \\
=&
\sum_{N\in \bZ} \sum_{\mu \in \cP} \Xi^{-N} \cdot
\Big\langle \mu \Big|
\Big( \Psi_{(\emptyset)} (q)  (-1)^{-2 L +2NC -N^2 } Q_1^{L-NC+N^2/2} \Big) \cdots \\
& \qquad\qquad\qquad\qquad \cdot
\Big( \Psi_{(\emptyset)}(q)
(-1)^{-2 L +2NC -N^2} Q_M^{L-NC+N^2/2} \Big)
 \Big| \mu \Big\rangle \\
 =& \sum_{N\in \bZ} \sum_{\mu \in \cP}
  \Xi^{-N} (-1)^{MN} Q^{N^2 /2}
\Big\langle \mu \Big|
\Big( \Psi_{(\emptyset)} (q)  (-1)^{-2 L } Q_1^{L} \Big) \cdots
\Big( \Psi_{(\emptyset)}(q)
(-1)^{-2 L } Q_M^{L} \Big)
 \Big| \mu \Big\rangle \\
= &
\Theta_3 \big( (-1)^M\Xi^{-1}; Q^{1/2} \big)\cdot
\sum_{\mu\in \cP }
\Big\langle \mu \Big|
\Big( \Psi_{(\emptyset)} (q)  (-1)^{-2 L } Q_1^{L} \Big) \cdots
\Big( \Psi_{(\emptyset)}(q)
(-1)^{-2 L } Q_M^{L} \Big)
 \Big| \mu \Big\rangle.
\end{split}
\end{equation*}
Here we have used the facts that $|\mu\rangle$, $\Psi_{(\emptyset)}(q)$ are of charge $0$
and $[C,L]=0$.
Notice that $L$ acts on
every basis vector $|\mu\rangle \in \cF^{(0)}$
by multiplying by an integer (see \eqref{eq-eigen-C&J}),
and thus the term $(-1)^{2L}$ in the above equality is actually the identity.
Therefore:
\be
\label{eq-Z-2-constpf}
\begin{split}
 Z^{\text{total}} (\bm 0; Q_i,-2; \Xi)
= \Theta_3 \big( (-1)^M\Xi^{-1}; Q^{1/2} \big)
\sum_{\mu\in \cP }
\Big\langle \mu \Big|
\Big( \Psi_{(\emptyset)} (q) Q_1^{L} \Big) \cdots
\Big( \Psi_{(\emptyset)}(q)
 Q_M^{L} \Big)
 \Big| \mu \Big\rangle.
\end{split}
\ee
By Corollary \ref{cor-comm-LPsi}, we have:
\begin{equation*}
\begin{split}
\Big( \Psi_{(\emptyset)} (q) Q_1^{L} \Big) \cdots
\Big( \Psi_{(\emptyset)}(q)
 Q_M^{L} \Big)
=& Q^L
\overset{\longrightarrow}{\prod_{l=1}^M}
\Big( \big(\prod_{a=l}^M Q_a \big)^{-L}
\cdot \Psi_{(\emptyset)}(q) \cdot \big(\prod_{a=l}^M Q_a \big)^{L}
\Big) \\
=& Q^L
\overset{\longrightarrow}{\prod_{l=1}^M} \Big(
\prod_{j\geq 0}^{\longleftarrow} \Gamma_{-} \big(\prod_{a=l}^M Q_a^{-1}  q^{-j-\half} \big)
\cdot \prod_{i\geq 0}^{\longrightarrow} \Gamma_{+} \big(\prod_{a=l}^M Q_a  q^{-i-\half} \big)
\Big).
\end{split}
\end{equation*}
To simplify the above operator,
we use \eqref{eq-commGamma} to move the operators $\Gamma_-$ to the left
and move $\Gamma_+$ to the right.
Notice here we also have $[\Gamma_-(z),\Gamma_-(w)] =[\Gamma_+(z),\Gamma_+(w)] = 0$.
Thus,
\begin{equation*}
\begin{split}
&\overset{\longrightarrow}{\prod_{l=1}^M} \Big(
\prod_{j\geq 0}^{\longleftarrow} \Gamma_{-} \big(\prod_{a=l}^M Q_a^{-1} \cdot q^{-j-\half} \big)
\cdot \prod_{i\geq 0}^{\longrightarrow} \Gamma_{+} \big(\prod_{a=l}^M Q_a \cdot q^{-i-\half} \big)
\Big)\\
=& c_1 \cdot
\Big( \prod_{l=1}^M \prod_{j\geq 0}\Gamma_{-} \big(\prod_{a=l}^M Q_a^{-1} \cdot q^{-j-\half} \big)
\Big) \Big(
\prod_{l=1}^M \prod_{i\geq 0} \Gamma_{+} \big(\prod_{a=l}^M Q_a \cdot q^{-i-\half} \big)
\Big)
\end{split}
\end{equation*}
where the coefficient $c_1$ is:
\begin{equation*}
\begin{split}
c_1 =& \prod_{1\leq k <l \leq M} \prod_{i,j\geq 0}
\frac{1}{1- (\prod_{a=k}^M Q_a) \cdot q^{-i-\half} \cdot (\prod_{b=l}^M Q_b^{-1} ) \cdot q^{-j-\half}} \\
= & \prod_{1\leq k <l \leq M}  \frac{1}{\cM (\prod_{i=k}^{l-1}Q_i ; q)}.
\end{split}
\end{equation*}
Then we apply \eqref{eq-comm-LGamma} to commute the operators $\Gamma_-$ with $Q^L$ and obtain:
\begin{equation*}
\begin{split}
&\Big( \Psi_{(\emptyset)} (q) Q_1^{L} \Big) \cdots
\Big( \Psi_{(\emptyset)}(q)
 Q_M^{L} \Big) \\
=& c_1 \cdot Q^L
\Big( \prod_{l=1}^M \prod_{j\geq 0}\Gamma_{-} \big(\prod_{a=l}^M Q_a^{-1} \cdot q^{-j-\half} \big)
\Big) \Big(
\prod_{l=1}^M \prod_{i\geq 0} \Gamma_{+} \big(\prod_{a=l}^M Q_a \cdot q^{-i-\half} \big)
\Big) \\
=& c_1 \cdot
\Big( \prod_{l=1}^M \prod_{j\geq 0}\Gamma_{-} \big(\prod_{a=1}^{l-1} Q_a \cdot q^{-j-\half} \big)
\Big) Q^L \Big(
\prod_{l=1}^M \prod_{i\geq 0} \Gamma_{+} \big(\prod_{a=l}^M Q_a \cdot q^{-i-\half} \big)
\Big).
\end{split}
\end{equation*}
In \cite[Lemma 3.2]{ya} the following identity has been proved for $|Q_i|<1$:
\begin{equation*}
\begin{split}
& \sum_{\mu\in \cP}
\Big\langle\mu \Big|
\Big( \prod_{l=1}^M \prod_{j\geq 0}\Gamma_{-} \big(\prod_{a=1}^{l-1} Q_a \cdot q^{-j-\half} \big)
\Big) Q^L \Big(
\prod_{l=1}^M \prod_{i\geq 0} \Gamma_{+} \big(\prod_{a=l}^M Q_a \cdot q^{-i-\half} \big)
\Big| \mu \Big\rangle
\Big) \\
=& \prod_{i=1}^\infty \frac{1}{1-Q^i} \cdot
\prod_{k,l= 1}^M \prod_{j=0}^\infty
\frac{1}{\cM(Q^j \cdot \prod_{a=1}^{k-1} Q_a \cdot \prod_{b=l}^M Q_i ;q)},
\end{split}
\end{equation*}
and thus:
\begin{equation*}
\begin{split}
&\sum_{\mu\in \cP }
\Big\langle \mu \Big|
\Big( \Psi_{(\emptyset)} (q) Q_1^{L} \Big) \cdots
\Big( \Psi_{(\emptyset)}(q)
 Q_M^{L} \Big)
 \Big| \mu \Big\rangle \\
 =&
\prod_{1\leq k <l \leq M}  \frac{1}{\cM (\prod_{i=k}^{l-1}Q_i ; q)} \cdot
\prod_{i=1}^\infty \frac{1}{1-Q^i} \cdot
\prod_{k,l= 1}^\infty \prod_{j=0}^\infty
\frac{1}{\cM(Q^j \cdot \prod_{a=1}^{k-1} Q_a \cdot \prod_{b=l}^M Q_i ;q)}.
\end{split}
\end{equation*}
Plug this into equation \eqref{eq-Z-2-constpf},
and then we obtain the conclusion of this theorem.
\end{proof}

\subsection{Affine coordinates of $Z^{\text{total}} (\textbf{t}; Q_i,-2; \Xi)$}

Now we compute the affine coordinates of the tau-function $Z^{\text{total}} (\mathbf{t}; Q_i,-2; \Xi)$
on the Sato Grassmannian.
The tau-function is completely determined by
the affine coordinates together with the constant term computed above.

\begin{Theorem}\label{prop:affine coord}
The affine coordinates $\{a_{n,m}^{(-2)}\}_{n,m\geq 0}$ of the tau-function
$Z^{\text{total}} (\mathbf{t}; Q_i,-2; \Xi)$ on the big cell of the Sato Grassmannian are:
\begin{equation*}
\begin{split}
a_{n,m}^{(-2)} =
& \frac{\Theta_3 \big( (-1)^M (q^{m+n+1}  \Xi)^{-1} ; Q^{\half} \big)}
{\Theta_3 \big( (-1)^M  \Xi^{-1}; Q^{\half} \big) }
\cdot
\frac{(-1)^{m+1} q^{\half m^2 +m +\half n +\half} }
{(1-q^{m+n+1}) \prod_{j=0}^{m-1} (1-q^{m-j}) \cdot \prod_{i=0}^{n-1}(1-q^{n-i}) } \\
&
\cdot \prod_{j= 0}^\infty \prod_{l=2}^M
\frac{  1-\prod_{a=1}^{l-1}Q_a \cdot q^{-n-j-1}}
{1-\prod_{a=1}^{l-1}Q_a \cdot q^{m-j}}
 \cdot \prod_{j=1}^\infty \frac{(1-Q^j)^2}{(1-Q^j q^{m+n+1})(1-Q^j q^{-m-n-1})}
\\
& \cdot
\prod_{j=0}^\infty \prod_{l=1}^M \prod_{i\geq 0} \frac{1-Q^j \prod_{b=l}^M Q_b \cdot q^{-m-i-1}}
{1-Q^j \prod_{b=l}^M Q_b \cdot q^{n-i}}
\cdot
\prod_{j=1}^\infty\prod_{k=1}^M \prod_{i\geq 0} \frac{1-Q^j  \prod_{a=1}^{k-1} Q_a \cdot q^{-n-i-1}}
{1-Q^j \prod_{a=1}^{k-1} Q_a \cdot q^{m-i}},
\end{split}
\end{equation*}
where $Q = Q_1Q_2\cdots Q_M$.
\end{Theorem}
\begin{proof}
From the discussions in \S \ref{sec-KPaffine} we know that the affine coordinates of the tau-function $Z^{\text{total}} (\mathbf{t}; Q_i,-2; \Xi)$ are given by
$a_{n,m}^{(-2)} = (-1)^n \cdot A_{m,n}^{(-2)}$, where
\begin{equation*}
\begin{split}
 A_{m,n}^{(-2)} = &
 \frac{1}{Z^{\text{total}} (\bm 0; Q_i,-2; \Xi)} \cdot
 \sum_{N\in \bZ} \sum_{\mu\in \cP}
\Big\langle \mu \Big|
R^N \Xi^C
\Big( \Psi_{(m|n)} (q)  (-1)^{-2  L} Q_1^L \Big)\cdot  \\
& \qquad\qquad \Big( \Psi_{(\emptyset)}(q)
(-1)^{-2 L} Q_2^L \Big) \cdots
\Big( \Psi_{(\emptyset)}(q)
(-1)^{-2 L} Q_M^L \Big)
R^{-N} \Big| \mu \Big\rangle.
\end{split}
\end{equation*}
We need to compute the vacuum expectation value in the above formula.
Similar to the case in Theorem \ref{prop-Z-2-const},
we apply Lemma \ref{lem-comm-RN} and Lemma \ref{lem-comm-RN-2} to obtain:
\begin{equation*}
\begin{split}
& \sum_{N\in \bZ} \sum_{\mu\in \cP}
\Big\langle \mu \Big|
R^N \Xi^C
\Big( \Psi_{(m|n)} (q)  (-1)^{-2  L} Q_1^L \Big)\cdot
\overset{\longrightarrow}{\prod_{i=2}^M}
 \Big( \Psi_{(\emptyset)}(q)
(-1)^{-2 L} Q_i^L \Big) \cdot
R^{-N} \Big| \mu \Big\rangle \\
=& \sum_{N\in \bZ} \sum_{\mu\in \cP} \Xi^{-N} q^{-N(m+n+1)} (-1)^{MN} Q^{N^2 /2}
\cdot
\Big\langle \mu \Big|
\Psi_{(m|n)} (q)   Q_1^L \cdot
\overset{\longrightarrow}{\prod_{i=2}^M}
 \Big( \Psi_{(\emptyset)}(q) Q_i^L \Big)\Big| \mu \Big\rangle \\
=& \Theta_3 \big( (-1)^M (q^{m+n+1} \cdot \Xi)^{-1} ; Q^{1/2} \big)
\cdot \sum_{\mu\in \cP}
\Big\langle \mu \Big|
\Psi_{(m|n)} (q)   Q_1^L \cdot
\overset{\longrightarrow}{\prod_{i=2}^M}
 \Big( \Psi_{(\emptyset)}(q) Q_i^L \Big)\Big| \mu \Big\rangle,
\end{split}
\end{equation*}
where
\be
\label{eq--2affine-pf2}
\begin{split}
& \Psi_{(m|n)} (q)   Q_1^L \cdot
\overset{\longrightarrow}{\prod_{i=2}^M}
 \Big( \Psi_{(\emptyset)}(q) Q_i^L \Big)\\
=& Q^L \cdot
\Big( Q^{-L} \Psi_{(m|n)}(q) Q^L \Big)\cdot
\overset{\longrightarrow}{\prod_{l=2}^M}
\Big( \big(\prod_{a=l}^M Q_a \big)^{-L}
\cdot \Psi_{(\emptyset)}(q) \cdot \big(\prod_{a=l}^M Q_a \big)^{L}
\Big).
\end{split}
\ee
By Corollary \ref{cor-comm-LPsi}, we have:
\begin{equation*}
\big(\prod_{a=l}^M Q_a \big)^{-L}
 \Psi_{(\emptyset)}(q)  \big(\prod_{a=l}^M Q_a \big)^{L}
\Big) =
\Big(
\prod_{j\geq 0}^{\longleftarrow} \Gamma_{-} \big(\prod_{a=l}^M Q_a^{-1}  q^{-j-\half} \big)
\cdot \prod_{i\geq 0}^{\longrightarrow} \Gamma_{+} \big(\prod_{a=l}^M Q_a  q^{-i-\half} \big)
\Big),
\end{equation*}
and
\be
\label{eq--2affine-pf1}
\begin{split}
& Q^{-L} \Psi_{(m|n)}(q) Q^L \\
=& f_1(q^{m+\half})\prod_{j>m} \Gamma_- (Q^{-1}q^{-j-\half}) \cdot
\Gamma_+ (Qq^{m+\half}) Q^C R^{-1} q^{(m+\half)C} \cdot
\prod_{j=0}^{m-1} \Gamma_- (Q^{-1}q^{-j-\half})  \\
& \cdot f_2(q^{n+\half})
\prod_{i=0}^{n-1} \Gamma_+ (Qq^{-i-\half}) \cdot
q^{(n+\half) C} R Q^{-C} \Gamma_-(Q^{-1}q^{n+\half}) \cdot
\prod_{i>n} \Gamma_+ (Qq^{-i-\half}).
\end{split}
\ee
Now we apply Lemma \ref{lem-comm-RN} and rewrite:
\begin{equation*}
\begin{split}
 R^{-1} q^{(m+\half)C} = q^{(m+\half)(C+1)} R^{-1},\qquad
 q^{(n+\half) C}  R = R q^{(n+\half) (C+1)},
\end{split}
\end{equation*}
and then apply the first equality in \eqref{eq-comm-RN-2}.
Thus  \eqref{eq--2affine-pf1} equals to:
\begin{equation*}
\begin{split}
& f_1(q^{m+\half})\prod_{j>m} \Gamma_- (Q^{-1}q^{-j-\half}) \cdot
\Gamma_+ (Qq^{m+\half}) Q^C q^{(m+\half) (C+1)} \cdot
\prod_{j=0}^{m-1} \Gamma_- (Q^{-1}q^{-j-\half})  \\
& \cdot f_2(q^{n+\half})
\prod_{i=0}^{n-1} \Gamma_+ (Qq^{-i-\half}) \cdot
q^{(n+\half) (C+1)} Q^{-C}  \Gamma_-(Q^{-1}q^{n+\half}) \cdot
\prod_{i>n} \Gamma_+ (Qq^{-i-\half}).
\end{split}
\end{equation*}
And since $\Gamma_\pm (z)$ and $|\mu\rangle$ are all of charge $0$,
we can simply replace $C$ by $0$ in this expression and rewrite \eqref{eq--2affine-pf1} as:
\begin{equation*}
\begin{split}
& q^{m+n+1} \cdot f_1(q^{m+\half})f_2(q^{n+\half})
\prod_{j>m} \Gamma_- (Q^{-1}q^{-j-\half}) \cdot
\Gamma_+ (Qq^{m+\half}) \cdot
\prod_{j=0}^{m-1} \Gamma_- (Q^{-1}q^{-j-\half})  \\
& \cdot
\prod_{i=0}^{n-1} \Gamma_+ (Qq^{-i-\half}) \cdot
  \Gamma_-(Q^{-1}q^{n+\half}) \cdot
\prod_{i>n} \Gamma_+ (Qq^{-i-\half}),
\end{split}
\end{equation*}
where
\begin{equation*}
q^{m+n+1} Q^{-1}f_1(q^{m+\half})f_2(q^{n+\half})
=  (-1)^{m+n+1} \cdot q^{\half m^2 +m +\half n +\half} \cdot Q^{-1}.
\end{equation*}
Plug these into \eqref{eq--2affine-pf2},
and we conclude that:
\be
\label{eq--2affine-pf3}
\begin{split}
& \sum_{\mu\in \cP}\Big\langle \mu \Big|
\Psi_{(m|n)} (q)   Q_1^L \cdot
\overset{\longrightarrow}{\prod_{i=2}^M}
 \Big( \Psi_{(\emptyset)}(q) Q_i^L \Big)\Big| \mu \Big\rangle \\
=& (-1)^{m+n+1} q^{\half m^2 +m +\half n +\half}
\sum_{\mu\in \cP}
\Big\langle \mu \Big| Q^L
\prod_{j>m} \Gamma_- (Q^{-1}q^{-j-\half}) \cdot
\Gamma_+ (Qq^{m+\half})  \\
& \quad \cdot \prod_{j=0}^{m-1} \Gamma_- (Q^{-1}q^{-j-\half})
 \cdot
\prod_{i=0}^{n-1} \Gamma_+ (Qq^{-i-\half}) \cdot
  \Gamma_-(Q^{-1}q^{n+\half}) \cdot
\prod_{i>n} \Gamma_+ (Qq^{-i-\half})  \\
& \quad \cdot \overset{\longrightarrow}{\prod_{l=2}^M}
\Big(
\prod_{j\geq 0}^{\longleftarrow} \Gamma_{-} \big(\prod_{a=l}^M Q_a^{-1} \cdot  q^{-j-\half} \big)
\cdot \prod_{i\geq 0}^{\longrightarrow} \Gamma_{+} \big(\prod_{a=l}^M Q_a \cdot  q^{-i-\half} \big)
\Big)
\Big| \mu \Big\rangle.
\end{split}
\ee
Again we use the commutation relation \eqref{eq-commGamma} to move the operators $\Gamma_-$ to the left
and $\Gamma_+$ to the right,
and in this way we can rewrite \eqref{eq--2affine-pf3} as:
\be
\label{eq--2affpr4}
\begin{split}
& c_2 \cdot \sum_{\mu\in \cP}
\Big\langle \mu\Big| Q^L \cdot
\Gamma_-(Q^{-1}q^{n+\half}) \Gamma_-(Q^{-1}q^{-m-\half})^{-1} \cdot
\prod_{k=1}^M \prod_{j\geq 0} \Gamma_{-} \big(\prod_{a=k}^M Q_a^{-1}\cdot  q^{-j-\half} \big) \\
& \qquad \cdot \prod_{l=1}^M \prod_{i\geq 0} \Gamma_{+} \big(\prod_{b=l}^M Q_b \cdot  q^{-i-\half} \big)
\cdot \Gamma_+ (Qq^{m+\half}) \cdot \Gamma_+(Qq^{-n-\half})^{-1}
\Big|\mu\Big\rangle,
\end{split}
\ee
where the coefficient $c_2$ is:
\begin{equation*}
\begin{split}
 c_2 =&
\frac{(-1)^{m+n+1} q^{\half m^2 +m +\half n +\half} }
{\prod_{j=0}^{m-1} (1-q^{m-j}) \cdot \prod_{i=0}^{n-1}(1-q^{n-i}) \cdot (1-q^{m+n+1})} \\
& \cdot \frac{1}{\prod_{1\leq k<l\leq M} \cM (\prod_{i=k}^{l-1}Q_i ;q)}
\cdot \frac{\prod_{l=2}^M \prod_{j\geq 0} (1-\prod_{a=1}^{l-1}Q_a \cdot q^{-n-j-1})}
{\prod_{l=2}^M \prod_{j\geq 0} (1-\prod_{a=1}^{l-1}Q_a \cdot q^{m-j})}.
\end{split}
\end{equation*}
Then by Lemma \ref{lem-comm-LGamma} we see that
\eqref{eq--2affpr4} equals to:
\be
\label{eq--2aff-pf5}
\begin{split}
& c_2 \cdot \sum_{\mu\in \cP}
\Big\langle \mu\Big|
\Gamma_-(q^{n+\half}) \Gamma_-(q^{-m-\half})^{-1} \cdot
\prod_{k=1}^M \prod_{j\geq 0} \Gamma_{-} \big(\prod_{a=1}^{k-1} Q_a \cdot  q^{-j-\half} \big) \\
& \qquad \cdot Q^L \cdot \prod_{l=1}^M \prod_{i\geq 0} \Gamma_{+} \big(\prod_{b=l}^M Q_b \cdot  q^{-i-\half} \big)
\cdot \Gamma_+ (Qq^{m+\half}) \Gamma_+(Qq^{-n-\half})^{-1}
\Big|\mu\Big\rangle.
\end{split}
\ee
This can be computed by exactly the same method used in the proof of \cite[Lemma 3.2]{ya},
and here we omit the details.
As a consequence,
\eqref{eq--2aff-pf5} equals to:
\begin{equation*}
\begin{split}
& c_2 \cdot \prod_{j=1}^\infty \frac{1-Q^j}{(1-Q^j q^{m+n+1})(1-Q^j q^{-m-n-1})}
\cdot
\prod_{j=0}^\infty \prod_{k,l=1}^M \frac{1}{\cM(Q^j  \prod_{b=l}^M Q_b \cdot \prod_{a=1}^{k-1} Q_a ; q)}
\\
&
\cdot
\prod_{j=0}^\infty \prod_{l=1}^M \prod_{i\geq 0} \frac{1-Q^j \prod_{b=l}^M Q_b \cdot q^{-m-i-1}}
{1-Q^j \prod_{b=l}^M Q_b \cdot q^{n-i}}
\cdot
\prod_{j=1}^\infty\prod_{k=1}^M \prod_{i\geq 0} \frac{1-Q^j  \prod_{a=1}^{k-1} Q_a \cdot q^{-n-i-1}}
{1-Q^j \prod_{a=1}^{k-1} Q_a \cdot q^{m-i}}.
\end{split}
\end{equation*}
Combining these computations with \eqref{eq-Z-2-const},  we complete the proof.
\end{proof}

\subsection{Quantum spectral curve for $Z^{\text{total}} (\textbf{t}; Q_i,-2; \Xi)$}

In this subsection,
we consider the principal specialization $t_k = \frac{ z^{-k}}{k}$ of the
tau-function $Z^{\text{total}} (\mathbf{t}; Q_i,-2; \Xi)$
and derive a difference equation for it.
Such a principal specialization of a tau-function of the KP hierarchy
is equivalent to the evaluation at $\mathbf{t} = \bm 0$ of the wave-function
associated with this tau-function \cite{djm},
which coincides with the first basis vector of the corresponding point on the Sato Grassmannian.
It is also the fermionic one-point function
from the viewpoint of boson-fermion correspondence.
In the literature,
this principal specialization of a tau-function together with a differential or difference operator annihilating it
is understood as the quantum spectral curve of
the corresponding model in theoretical or mathematical physics,
see e.g. \cite{al, gs}.

Now in our case, the principle specialization gives
\begin{align*}
	\Psi(z)
	=\frac{Z^{\text{total}} (\mathbf{t}; Q_i,-2; \Xi)|_{t_k = z^{-k}/k}}
	{Z^{\text{total}} (\bm 0; Q_i,-2; \Xi)}.
\end{align*}
Moreover,
from the KP integrability of the tau-function $Z^{\text{total}} (\mathbf{t}; Q_i,-2; \Xi)$,
the $\Psi(z)$ is also equal to the first basis vector of the point on the Sato Grassmannian
corresponding to this tau-function
and is also called the fermionic one-point function.
Thus, it can be also represented as the following Laurent series:
\be
\begin{split}
\Psi (z) =&
1+ \sum_{m\geq 0}  a_{0,m}^{(-2)} \cdot z^{-m-1},
\end{split}
\ee
where, from Theorem \ref{prop:affine coord},
\begin{equation*}
\begin{split}
a_{0,m}^{(-2)} =
& \frac{\Theta_3 \big( (-1)^M (q^{m+1}  \Xi)^{-1} ; Q^{\half} \big)}
{\Theta_3 \big( (-1)^M  \Xi^{-1}; Q^{\half} \big) }
\cdot
\frac{(-1)^{m+1} q^{\half m^2 +m +\half} }
{\prod_{j=1}^{m+1} (1-q^j) }
\cdot \prod_{l=2}^M
\frac{ 1 }
{\prod_{i=0}^m ( 1-\prod_{a=1}^{l-1}Q_a \cdot q^{i}) } \\
& \cdot \prod_{j=0}^\infty \frac{(1-Q^{j+1})^2 \cdot
(1-Q^{j+1} q^{m+1})^{-1} \cdot (1-Q^{j+1} q^{-m-1})^{-1}}
{\prod_{l=1}^M  \prod_{i=0}^m ( 1-Q^j \prod_{b=l}^M Q_b \cdot q^{-i} ) \cdot
\prod_{k=1}^M \prod_{i=0}^m  ( 1-Q^{j+1} \prod_{a=1}^{k-1} Q_a \cdot q^{i}) }.
\end{split}
\end{equation*}
We then derive the following difference equation satisfied by $\Psi(z)$,
which is understood as a quantum spectral curve for this model:
\begin{Proposition}
Let $\wP$  be the following $q$-difference operator:
\be
\label{eq-def-oprP}
\begin{split}
	\wP =
	&  \Theta_3 \big( (-1)^M  \Xi^{-1} q^{z\pd_z -1} ; Q^{\half} \big)
	\cdot
	q^{-z\pd_z +\half}
	\cdot \prod_{j=0}^\infty
	(1-Q^{j+1} q^{-z\pd_z}) (1-Q^{j+1} q^{z\pd_z}) \\
	&+  z \cdot \Theta_3 \big( (-1)^M  \Xi^{-1} q^{z\pd z +1 } ; Q^{\half} \big)
	\cdot
	( 1-q^{-z\pd_z } )
	\cdot \prod_{l=2}^M
	( 1-\prod_{a=1}^{l-1}Q_a \cdot q^{-z\pd_z-1} ) \\
	& \cdot
	\prod_{j=0}^\infty \prod_{l=1}^M   ( 1-Q^j \prod_{b=l}^M Q_b \cdot q^{z\pd_z +1} ) \cdot
	\prod_{j=0}^\infty  \prod_{k=1}^M   ( 1-Q^{j+1} \prod_{a=1}^{k-1} Q_a \cdot q^{-z\pd_z -1}) \\
	& \cdot \prod_{j=0}^\infty
	(1-Q^{j+1} q^{-z\pd_z}) \cdot (1-Q^{j+1} q^{z\pd_z}).
\end{split}
\ee
The fermionic one-point function
$\Psi(z)$ of the tau-function $Z^{\text{total}} (\textbf{t}; Q_i,-2; \Xi)$
satisfies the following difference equation:
\be
\label{eq-Z-2-qsc}
\wP \Psi(z) =0.
\ee
\end{Proposition}
\begin{proof}
Notice that for every $m\geq 0$ one has:
\begin{equation*}
\begin{split}
\frac{a_{0,m+1}^{(-2)} }{a_{0,m}^{(-2)}}=
& - \frac{\Theta_3 \big( (-1)^M (q^{m+2}  \Xi)^{-1} ; Q^{\half} \big)}
{\Theta_3 \big( (-1)^M (q^{m+1}  \Xi)^{-1} ; Q^{\half} \big)}
\cdot
\frac{q^{m +\frac{3}{2}} }
{1-q^{m+2} }
\cdot \prod_{l=2}^M
\frac{ 1 }
{ 1-\prod_{a=1}^{l-1}Q_a \cdot q^{m+1} } \\
& \cdot \prod_{j=0}^\infty \frac{ 1 }
{\prod_{l=1}^M   ( 1-Q^j \prod_{b=l}^M Q_b \cdot q^{-m-1} ) \cdot
\prod_{k=1}^M   ( 1-Q^{j+1} \prod_{a=1}^{k-1} Q_a \cdot q^{m+1}) } \\
& \cdot \prod_{j=0}^\infty
\frac{(1-Q^{j+1} q^{m+1}) \cdot (1-Q^{j+1} q^{-m-1})}
{(1-Q^{j+1} q^{m+2}) \cdot (1-Q^{j+1} q^{-m-2})}.
\end{split}
\end{equation*}
Now denote by $\wB$, $\wC$ the following operators:
\begin{equation*}
\begin{split}
\wB =
&  \Theta_3 \big( (-1)^M  \Xi^{-1} q^{z\pd_z -1} ; Q^{\half} \big)
\cdot
q^{-z\pd_z +\half}
\cdot \prod_{j=0}^\infty
(1-Q^{j+1} q^{-z\pd_z}) (1-Q^{j+1} q^{z\pd_z}) ,
\end{split}
\end{equation*}
and
\begin{equation*}
\begin{split}
\wC =
&  \Theta_3 \big( (-1)^M  \Xi^{-1} q^{z\pd z} ; Q^{\half} \big)
\cdot
( 1-q^{-z\pd_z +1} )
\cdot \prod_{l=2}^M
( 1-\prod_{a=1}^{l-1}Q_a \cdot q^{-z\pd_z} ) \\
& \cdot
\prod_{j=0}^\infty \prod_{l=1}^M   ( 1-Q^j \prod_{b=l}^M Q_b \cdot q^{z\pd_z} ) \cdot
\prod_{j=0}^\infty \prod_{k=1}^M   ( 1-Q^{j+1} \prod_{a=1}^{k-1} Q_a \cdot q^{-z\pd_z}) \\
& \cdot \prod_{j=0}^\infty
 (1-Q^{j+1} q^{-z\pd_z+1}) \cdot (1-Q^{j+1} q^{z\pd_z-1}).
\end{split}
\end{equation*}
Then for every $m\geq 0$,
\begin{equation*}
\begin{split}
\wB (z^{-m-1}) =
 \Theta_3 \big( (-1)^M  (\Xi q^{m+2})^{-1} ; Q^{\half} \big)
q^{m+\frac{3}{2} }
 \prod_{j=0}^\infty
(1-Q^{j+1} q^{+1}) (1-Q^{j+1} q^{-m-1})
\cdot z^{-m-1}
\end{split}
\end{equation*}
is a constant multiple of $z^{-m-1}$,
and thus
\begin{equation*}
\begin{split}
&\wC^{-1} \wB (z^{-m-1}) \\ =&
 \Theta_3 \big( (-1)^M  (\Xi q^{m+2})^{-1} ; Q^{\half} \big)
q^{m+\frac{3}{2} } \cdot
 \prod_{j=0}^\infty
(1-Q^{j+1} q^{+1}) (1-Q^{j+1} q^{-m-1})
\cdot \wB^{-1}( z^{-m-1}) \\
=& - \frac{a_{0,m+1}^{(-2)} }{a_{0,m}^{(-2)}} \cdot z^{-m-1}.
\end{split}
\end{equation*}
Moreover,
one has $\wC^{-1} \wB (1) = - a_{0,0}^{(-2)}$,
and thus:
\begin{equation*}
\wC^{-1} \wB \big(\Psi(z)\big) = -a_{0,0}^{(-2)} - \sum_{m\geq 0} a_{0,m+1}^{(-2)} z^{-m-1}
 = z - z \cdot \Psi(z),
\end{equation*}
which implies:
\begin{equation*}
\wC (\wC^{-1}\wB + z) \Psi(z) = \wC(z) =0.
\end{equation*}
Here $\wC(z) = 0$ since $( 1-q^{-z\pd_z +1} )(z)  = z-z =0$.
Therefore the conclusion follows from the fact $\wP = \wB + \wC\circ z$
as operators acting on formal series in $z$.
\end{proof}

\begin{Example}
When $M=1$ and $Q_1 = 0$,
the fermionic one-point function $\Psi(z)$ becomes the quantum dilogarithm
\ben
\Psi(z) = \sum_{n \geq 0} \frac{(-1)^n\cdot q^{n^2/2}}{\prod_{j=1}^n (1-q^j)} z^{-n} ,
\een
and the $q$-difference equation \eqref{eq-Z-2-qsc} reduces to:
\begin{equation*}
\Big( q^{-z\pd_z + \half} + z(1-q^{-z \pd_z}) \Big) \Psi(z) =0.
\end{equation*}
This is equivalent to:
\begin{equation}
\Big( q^{-\half} z^{-1} + q^{z\pd_z} -1 \Big) \Psi(z) =0,
\end{equation}
which matches with the discussions for the one-legged topological vertex in \cite[\S 5.10]{adkmv}.
\end{Example}

\section{Example: the Local $\mathbb{P}^2$ Model}
\label{sec:P2}
In this section,
we use the topological vertex and our explicit construction in \S \ref{sec:def Ztotal}
to compute some leading terms of the total partition function of the local $\mathbb{P}^2$ model.
We use them to check the simplest P\"ucker relations for the coefficients of the Schur expansion of the total partition function.

For the local $\mathbb{P}^2$ model,
we have $M=3$ and $\gamma_1=\gamma_2=\gamma_3=1$,
and the K\"ahler parameters are $Q_1=Q_2=Q_3=Q$.
Then by the gluing rule \cite{akmv},
the open string partition function of the local $\mathbb{P}^2$ model is:
\begin{equation*}
	\begin{split}\label{eqn:ZP2}
		Z^{\mathbb{P}^2} (\mathbf{t};Q)
		= \sum_{\lambda \in \cP} \sum_{\mu^1,\mu^2, \mu^3 \in \cP}
		(-Q)^{\sum_{i=1}^3 |\mu^i|}
		q^{\sum_{i=1}^3 \kappa_{\mu^i}}
		\cdot
		W_{\mu^{3,t},\lambda,\mu^1} W_{\mu^{1,t},\emptyset,\mu^2}
		W_{\mu^{2,t},\emptyset,\mu^3} \cdot s_\lambda(\mathbf{t}).
	\end{split}
\end{equation*}
Consider the coefficients of its Schur function expansion:
\begin{equation*}
	\begin{split}
		Z_{\lambda}^{\mathbb{P}^2} (Q) =  \sum_{\mu^1,\mu^2, \mu^3 \in \cP}
		(-Q)^{\sum_{i=1}^3 |\mu^i|}
		q^{\sum_{i=1}^3 \kappa_{\mu^i}}
		\cdot
		W_{\mu^{3,t},\lambda,\mu^1} W_{\mu^{1,t},\emptyset,\mu^2}
		 W_{\mu^{2,t},\emptyset,\mu^3}.
	\end{split}
\end{equation*}
One can verify that the coefficients $\{Z_{\lambda}^{\mathbb{P}^2} (Q)\}_{\lambda\in\cP}$
do not satisfy Pl\"ucker relation
since the open string partition function $Z^{\mathbb{P}^2} (\mathbf{t};Q)$ itself
is not a tau-function of the KP hierarchy (see the discussion in \cite[\S 5.13]{adkmv}).

From the construction in \S \ref{sec:def Ztotal} of this paper,
the total partition function of the local $\mathbb{P}^2$ model is given by
\be\label{eqn:P2 total}
Z^{\mathbb{P}^2,\text{total}} (\mathbf{t};Q) =
\sum_{N\in \bZ} Z^{\mathbb{P}^2,(N)}(\mathbf{t}; Q)\cdot \Xi^{-N}.
\ee
For convenience,
for each $N\in\mathbb{Z}$
we expand the $N$-component part $Z^{\mathbb{P}^2,(N)}(\mathbf{t}; Q_i,\gamma_i)$ in terms of Schur functions as
\be
Z^{\mathbb{P}^2,(N)}(\mathbf{t};Q_i,\gamma_i)
= \sum_{\lambda \in \cP}
Z^{\mathbb{P}^2,(N)}_\lambda (Q_i,\gamma_i) \cdot s_\lambda(\mathbf{t}).
\ee
Then from equation \eqref{eq-defZN-2},
we have
\be\label{eqn:Znla from Zla}
Z^{\mathbb{P}^2,(N)}_\lambda (\mathbf{t};Q_i,\gamma_i)
=q^{-N| \lambda |}  Q^{\frac{N^2}{2}} (-1)^{\frac{3N^2}{2}}
q^{\frac{3N(4N^2 -1)}{8}} \cdot
Z^{\mathbb{P}^2}_\lambda(q^{3N}Q).
\ee

In what follows we list some data for this model,
and these data are useful when we verify the Pl\"ucker equation of this model at the end of this section.
We will use the notation:
\begin{align}
	[k] =q^{k/2}-q^{-k/2}.
\end{align}

\begin{Example}\label{ex:0}
	For the constant part of the open string partition function, i.e. the part corresponding to $\lambda=\emptyset$,
	we have:
	\begin{align*}
		Z_{\emptyset}^{\mathbb{P}^2,(0)} (Q)
		=&1
		+Q\cdot(-3[1]^{-2})
		+Q^2\cdot(-\frac{3}{2}[2]^{-2}+\frac{9}{2}[1]^{-4}+6[1]^{-2})\\
		&+Q^3\cdot(-[3]^{-2}+\frac{9}{2}[2]^{-2}[1]^{-2}-\frac{9}{2}[1]^{-6}-18[1]^{-4}-27[1]^{-2}-10)+O(Q^4).
	\end{align*}
	Then from equation \eqref{eqn:Znla from Zla},
	for $N=1,2$, we have
	\begin{align*}
		& Z_{\emptyset}^{\mathbb{P}^2,(1)} (Q)
		=-\sqrt{-1} Q^{1/2} q^{9/8}
		\cdot Z_{\emptyset}^{\mathbb{P}^2,(0)} (q^3 Q),\\
		& Z_{\emptyset}^{\mathbb{P}^2,(2)} (Q)
		=q^{45/4} Q^2
		\cdot Z_{\emptyset}^{\mathbb{P}^2,(0)} (q^6 Q).
	\end{align*}
	For $N=-1$,
	we have
	\begin{align*}
	Z_{\emptyset}^{\mathbb{P}^2,(-1)} (Q)
	=-\sqrt{-1} Q^{1/2} q^{-9/8}
	\cdot Z_{\emptyset}^{\mathbb{P}^2,(0)} (q^{-3} Q).
\end{align*}
\end{Example}

\begin{Example}\label{ex:1}
	For the degree $1$ part,
	there is only one partition $\lambda=(1)$.
	We have
	\begin{align*}
		Z_{(1)}^{\mathbb{P}^2,(0)} (Q)
		=&[1]^{-1}
		+Q\cdot(-2[1]^{-1}-3[1]^{-3})\\
		&+Q^2\cdot(-\frac{3}{2}[2]^{-2}[1]^{-1}+\frac{9}{2}[1]^{-5}+12[1]^{-3}+5[1]^{-1})+O(Q^3).
	\end{align*}
	Then from equation \eqref{eqn:Znla from Zla},
	for $N=\pm1$, we have
	\begin{align*}
		Z_{(1)}^{\mathbb{P}^2,(1)} (Q)
		=&-\sqrt{-1} Q^{1/2} q^{1/8}
		\cdot Z_{(1)}^{\mathbb{P}^2,(0)} (q^3 Q),\\
		Z_{(-1)}^{\mathbb{P}^2,(1)} (Q)
		=&-\sqrt{-1} Q^{1/2} q^{-1/8}
		\cdot Z_{(1)}^{\mathbb{P}^2,(0)} (q^{-3} Q).
	\end{align*}
\end{Example}

\begin{Example}\label{ex:2}
	For the degree $2$ part,
	there are two partitions $\lambda=(2), (1,1)$.
	For the first case $\lambda=(2)$,
	we have
	\begin{align*}
		Z_{(2)}^{\mathbb{P}^2,(0)} (Q)
		=&\frac{1}{2}[2]^{-1}+\frac{1}{2}[1]^{-2}
		+Q\cdot\Big([2]^{-1}\big(-\frac{3}{2}[1]^{-2}-2-\frac{1}{2}[1]^{2}\big)
		-\frac{3}{2}[1]^{-4}-2[1]^{-2}-\frac{1}{2}\Big)\\
		&+Q^2\cdot\Big(-\frac{3}{4}[2]^{-3}-\frac{3}{4}[2]^{-2}[1]^{-2}
		+2[2]^{-1}\big(\frac{9}{4}[1]^{-4}+9[1]^{-2}+\frac{17}{2}+[1]^2\big)\\
		&\qquad+\frac{9}{4}[1]^{-6}+9[1]^{-4}+\frac{17}{2}[1]^{-2}+2\Big)
		+O(Q^3).
	\end{align*}
	Then from equation \eqref{eqn:Znla from Zla},
	for $N=\pm1$, we have
	\begin{align*}
		& Z_{(2)}^{\mathbb{P}^2,(1)} (Q)
		=-\sqrt{-1} Q^{1/2} q^{-7/8}
		\cdot Z_{(2)}^{\mathbb{P}^2,(0)} (q^3 Q),\\
		& Z_{(2)}^{\mathbb{P}^2,(-1)} (Q)
		=-\sqrt{-1} Q^{1/2} q^{7/8}
		\cdot Z_{(2)}^{\mathbb{P}^2,(0)} (q^{-3} Q).
	\end{align*}
	
	For the second case $\lambda=(1,1)$,
	we have
	\begin{align*}
		Z_{(1,1)}^{\mathbb{P}^2,(0)} (Q)
		=&-\frac{1}{2}[2]^{-1}+\frac{1}{2}[1]^{-2}
		+Q\cdot\Big([2]^{-1}\big(\frac{3}{2}[1]^{-2}+2+\frac{1}{2}[1]^{2}\big)
		-\frac{3}{2}[1]^{-4}-2[1]^{-2}-\frac{1}{2}\Big)\\
		&+Q^2\cdot\Big(\frac{3}{4}[2]^{-3}-\frac{3}{4}[2]^{-2}[1]^{-2}
		+[2]^{-1}\big(-\frac{9}{4}[1]^{-4}-9[1]^{-2}-\frac{17}{2}-2[1]^2\big)\\
		&\qquad+\frac{9}{4}[1]^{-6}+9[1]^{-4}+\frac{17}{2}[1]^{-2}+2\Big)
		+O(Q^3).
	\end{align*}
	Then from equation \eqref{eqn:Znla from Zla},
	for $N=\pm1$, we have
	\begin{align*}
		& Z_{(1,1)}^{\mathbb{P}^2,(1)} (Q)
		=-\sqrt{-1} Q^{1/2} q^{-7/8}
		\cdot Z_{(1,1)}^{\mathbb{P}^2,(0)} (q^3 Q),\\
		& Z_{(1,1)}^{\mathbb{P}^2,(-1)} (Q)
		=-\sqrt{-1} Q^{1/2} q^{7/8}
		\cdot Z_{(1,1)}^{\mathbb{P}^2,(0)} (q^{-3} Q).
	\end{align*}
\end{Example}

\begin{Example}\label{ex:3}
	For the degree 3 part,
	consider the case $\lambda=(2,1)$.
	We have
	\begin{align*}
		Z_{(2,1)}^{\mathbb{P}^2,(0)} (Q)
		=&-\frac{1}{3}[3]^{-1}
		+\frac{1}{3}[1]^{-3}\\
		&+Q\cdot\Big([3]^{-1}\big([1]^{-2}+2+\frac{2}{3}[1]^2\big)-[1]^{-5}-2[1]^{-3}-\frac{2}{3}[1]^{-1}\Big)
		+O(Q^2).
	\end{align*}
	Then from equation \eqref{eqn:Znla from Zla},
	for $N=\pm1$, we have
	\begin{align*}
		& Z_{(2,1)}^{\mathbb{P}^2,(1)} (Q)
		=-\sqrt{-1} Q^{1/2} q^{-15/8}
		\cdot Z_{(2,1)}^{\mathbb{P}^2,(0)} (q^3 Q),\\
		& Z_{(2,1)}^{\mathbb{P}^2,(-1)} (Q)
		=-\sqrt{-1} Q^{1/2} q^{15/8}
		\cdot Z_{(2,1)}^{\mathbb{P}^2,(0)} (q^{-3} Q).
	\end{align*}
\end{Example}

\begin{Example}\label{ex:4}
	For the degree $4$ part,
	we consider the case $\lambda=(2,2)$,
	which will appear in the first nontrivial Pl\"ucker relation
the total partition function of the local $\mathbb{P}^2$ model satisfies.
We have:
	\begin{align*}
		Z_{(2,2)}^{\mathbb{P}^2,(0)} (Q)
		=[3]^{-1}[2]^{-2}[1]^{-1}
		+Q\cdot\big([3]^{-1}[2]^{-2}[1]^{-1}
		\cdot(-3[1]^{-2}-8-2[1]^2)\big)
		+O(Q^2).
	\end{align*}
	Then from equation \eqref{eqn:Znla from Zla},
	for $N=\pm1$, we have
	\begin{align*}
		& Z_{(2,2)}^{\mathbb{P}^2,(1)} (Q)
		=-\sqrt{-1} Q^{1/2} q^{-23/8}
		\cdot Z_{(2,2)}^{\mathbb{P}^2,(0)} (q^3 Q),\\
		& Z_{(2,2)}^{\mathbb{P}^2,(-1)} (Q)
		=-\sqrt{-1} Q^{1/2} q^{23/8}
		\cdot Z_{(2,2)}^{\mathbb{P}^2,(0)} (q^{-3} Q).
	\end{align*}
\end{Example}

Using the above data
we now verify the first Pl\"ucker relation of the total partition function of the local $\mathbb{P}^2$ model.
Recall that
for a tau-function of the KP hierarchy with the following Schur function expansion
\begin{align*}
	Z(\mathbf{t})
	=\sum_{\lambda\in\mathcal{P}}
	c_{\lambda} \cdot s_{\lambda}(\mathbf{t}),
\end{align*}
the first nontrivial equation for these coefficients $\{c_{\lambda}\}_{\lambda\in\mathcal{P}}$ reads
\begin{align}\label{eqn:1st H}
	c_{(2,2)}\cdot c_{\emptyset}=
	\left|\begin{matrix}
		c_{(2,1)} & c_{(2)}\\
		c_{(1,1)} & c_{(0,0)}
	\end{matrix}
	\right|.
\end{align}
Now for convenience,
we expand the $\lambda$-part of the total partition function of this model as:
\begin{align*}
	Z^{\mathbb{P}^2,\text{total}}_{\lambda} (Q) =
	\sum_{N\in \bZ} Z^{\mathbb{P}^2,(N)}_{\lambda}(Q)\cdot \Xi^{-N},
\end{align*}
where $Z^{\mathbb{P}^2,\text{total}}_{\lambda} (Q)$ is
the coefficient of $s_{\lambda}(\mathbf{t})$ in the total partition function of the local $\mathbb{P}^2$ model in \eqref{eqn:P2 total}.
For the right-hand side of \eqref{eqn:1st H},
by Examples \ref{ex:1}-\ref{ex:3} we have:
\begin{align}\label{eqn:H r}
	\begin{split}
	&\left|\begin{matrix}
		Z_{(2,1)}^{\mathbb{P}^2,\text{total}} (Q) & Z_{(2)}^{\mathbb{P}^2,\text{total}} (Q)\\
		Z_{(1,1)}^{\mathbb{P}^2,\text{total}} (Q) & Z_{(1)}^{\mathbb{P}^2,\text{total}} (Q)
	\end{matrix}
	\right|\\
	=&\frac{q^{4}}{(q-1)^{4} (q^{2}+q+1) (q+1)^{2}}
	-\sqrt{-1}Q^{1/2}\cdot\frac{(q^{4}+1)\big(q^{\frac{9}{8}}\Xi^{-1}+q^{\frac{23}{8}}\Xi\big)}{(q+1)^{2} (q-1)^{4} (q^{2}+q+1)}\\
	&-Q \Big(\frac{q^{\frac{9}{4}}\Xi^{-2}+q^{\frac{23}{4}}\Xi^{2}}{(q+1)^{2} (q-1)^{4} (q^{2}+q+1)} +\frac{q^{10}-2 q^{9}+q^{8}+2 q^{7}+2 q^{5}+2 q^{3}+q^{2}-2 q+1}{(q-1)^{6} (q^{2}+q+1) (q+1)^{2}}\Big)\\
	&+\sqrt{-1}Q^{3/2}\cdot
	\frac{(3 q^{4}-4 q^{3}+4 q^{2}-4 q+3)\big(q^{\frac{17}{8}}\Xi^{-1}+q^{\frac{7}{8}}\Xi\big)}{(q+1) (q-1)^{6}}
	+O(Q^2).
	\end{split}
\end{align}
And for the left-hand side of \eqref{eqn:1st H},
by Examples \ref{ex:0} and \ref{ex:4} we have:
\begin{align}
	\begin{split}
	&Z_{(2,2)}^{\mathbb{P}^2,\text{total}} (Q)
	\cdot Z_{\emptyset}^{\mathbb{P}^2,\text{total}} (Q)\\
	=&\frac{q^{4}}{(q-1)^{4} (q^{2}+q+1) (q+1)^{2}}
	-\sqrt{-1}Q^{1/2}\cdot\frac{(q^{4}+1)\big(q^{\frac{9}{8}}\Xi^{-1}+q^{\frac{23}{8}}\Xi\big)}{(q+1)^{2} (q-1)^{4} (q^{2}+q+1)}\\
	&-Q \Big(\frac{q^{\frac{9}{4}}\Xi^{-2}+q^{\frac{23}{4}}\Xi^{2}}{(q+1)^{2} (q-1)^{4} (q^{2}+q+1)} +\frac{q^{10}-2 q^{9}+q^{8}+2 q^{7}+2 q^{5}+2 q^{3}+q^{2}-2 q+1}{(q-1)^{6} (q^{2}+q+1) (q+1)^{2}}\Big)\\
	&+\sqrt{-1}Q^{3/2}\cdot
	\frac{(3 q^{4}-4 q^{3}+4 q^{2}-4 q+3)\big(q^{\frac{17}{8}}\Xi^{-1}+q^{\frac{7}{8}}\Xi\big)}{(q+1) (q-1)^{6}}
	+O(Q^2),
	\end{split}
\end{align}
which matches with \eqref{eqn:H r} and verifies the leading terms of the equation \eqref{eqn:1st H}.

\vspace{.2in}

{\em Acknowledgements}.
The authors thank Prof. Qingsheng Zhang and Prof. Huaiqing Zuo for helpful discussions.
The first and second authors thank Prof. Huijun Fan, Prof. Xiaobo Liu, and Prof. Xiangyu Zhou for encouragement.
The second author is partly supported by the NSFC grant (No. 12401079).
The third author is partly supported by the NSFC grants (No. 12371254, 11890662, 12061131014).

\end{document}